\documentclass{LMCS} 
\usepackage{enumerate,hyperref}
\usepackage{amssymb}
\usepackage{latexsym}

\theoremstyle{plain}
\newtheorem{fewtheorem}{Theorem}
\newtheorem{proposition}[thm]{Proposition}
\newtheorem{theorem}[thm]{Theorem}
\newtheorem{lemma}[thm]{Lemma} 
\newtheorem{corollary}[thm]{Corollary}  

\theoremstyle{definition}
\newtheorem{convention}[thm]{Convention} 
\newtheorem{rema}[thm]{Remark} 
\newtheorem{exam}[thm]{Example}

\def\mvsquareforeod{\hbox{$\lhd$}}
\def\mveod{\ifmmode\mvsquareforqed\else{\unskip\nobreak\hfil
\penalty50\hskip1em\null\nobreak\hfil\mvsquareforeod
\parfillskip=0pt\finalhyphendemerits=0\endgraf}\fi}

\newenvironment{definition}{\begin{defi}\rm}{\hfill$\lhd$\end{defi}}
\newenvironment{remark}{\begin{rema}\rm}{\end{rema}}
\newenvironment{example}{\begin{exam}\rm}{\end{exam}}

\def\mvsquareforqed{\hbox{\textsc{qed}}}
\def\mvqed{\ifmmode\mvsquareforqed\else{\unskip\nobreak\hfil
\penalty50\hskip1em\null\nobreak\hfil\mvsquareforqed
\parfillskip=0pt\finalhyphendemerits=0\endgraf}\fi}

\newenvironment{proofof}[1]{\begin{trivlist}\item[\hskip\labelsep{\bf
Proof~of~{#1}.\ }]}{\hspace*{\fill} \qed\end{trivlist}}

\newtheorem{claim2}{Claim}

\newenvironment{claim}{\begin{claim2}\rm}{\end{claim2}\rm}
\newenvironment{claimeerst}{\setcounter{claim2}{0}
               \begin{claim2}\rm}{\end{claim2}\rm}

\newenvironment{pfclaim}{\begin{trivlist}\item[]{\sc Proof of
Claim}}{\hfill {\mbox{$\lhd$}}\end{trivlist}}

\newcommand{\emdef}[1]{\textsl{#1}}
\newcommand{\verberg}[1]{}

\newenvironment{tbs}{%
   \begin{itemize}\tt }{\end{itemize}}
\newcommand{\btbs}{\begin{tbs}}                                                                      
\newcommand{\etbs}{\end{tbs}}

\newcommand{\struc}[1]{\langle #1 \rangle}

\newcommand{\bbA}{\mathbb{A}}
\newcommand{\bbB}{\mathbb{B}}

\newcommand{\bbN}{\mathbb{N}}

\newcommand{\bbS}{\mathbb{S}}

\newcommand{\bis}{\mathrel{\mathchoice
{\raisebox{.3ex}{$\,
  \underline{\makebox[.7em]{$\leftrightarrow$}}\,$}}
{\raisebox{.3ex}{$\,
  \underline{\makebox[.7em]{$\leftrightarrow$}}\,$}}
{\raisebox{.2ex}{$\,
  \underline{\makebox[.5em]{\scriptsize$\leftrightarrow$}}\,$}}
{\raisebox{.2ex}{$\,
  \underline{\makebox[.5em]{\scriptsize$\leftrightarrow$}}\,$}}}}

\newcommand{\class}[1]{\mathsf{#1}}

\newcommand{\clL}{\class{L}}

\newcommand{\Set}{\mathsf{Set}}
\newcommand{\Rel}{\mathsf{Rel}}

\newcommand{\BT}{\mathsf{B}} 
\newcommand{\F}{\mathsf{F}} 
\newcommand{\FG}{\mathsf{G}} 
\newcommand{\K}{\mathsf{P}} 
\newcommand{\Fb}{\overline{\F}} 
\newcommand{\Id}{\mathsf{Id}}
\newcommand{\id}{\mathit{id}}

\newcommand{\Oop}{\mathcal{O}}

\renewcommand{\phi}{\varphi} 

\newcommand{\Base}{\mathit{Base}}

\newcommand{\B}{\mathcal{B}}

\newcommand{\G}{\mathcal{G}}
\newcommand{\eloi}{\exists}
\newcommand{\abel}{\forall}
\newcommand{\Win}{\mathit{Win}}
\newcommand{\Ref}{\mathit{Ref}}

\newcommand{\ai}{a_{I}}
\newcommand{\Ri}{R_{I}}
\newcommand{\Acc}{\mathit{Acc}}
\newcommand{\Par}{\Omega}

\newcommand{\sh}[1]{{#1}^{\sharp}}
\newcommand{\NOT}{\mathrm{NBT}}
\newcommand{\bl}{\bullet}

\newcommand{\bpr}{{\upharpoonright}}
\newcommand{\rst}[1]{\!\upharpoonright_{#1}\,}

\newcommand{\power}{\mathcal{P}}
\newcommand{\pwE}{\power_{\eloi}}
\newcommand{\pwA}{\power_{\abel}}
\newcommand{\nada}{\varnothing}

\newcommand{\Sb}[1]{\mathcal{P}(#1)}
\newcommand{\sse}{\subseteq}

\newcommand{\ovl}[1]{\overline{#1}}
\newcommand{\udl}[1]{\underline{#1}}

\newcommand{\Ran}{\mathsf{rng}}

\newcommand{\relA}{\sh{A}}
\newcommand{\Graph}[1]{\mathit{Gr}(#1)}
\newcommand{\converse}[1]{#1\breve{\hspace*{1mm}}}
\newcommand{\ev}[1]{\mathit{ev}_{#1}}
\newcommand{\eva}{\ev{a}}

\newcommand{\inA}{{\in_{A}}}
\newcommand{\niA}{{\ni_{A}}}

\newcommand{\al}{\alpha}
\newcommand{\ga}{\gamma}
\newcommand{\de}{\delta}
\newcommand{\la}{\lambda}
\newcommand{\si}{\sigma}
\newcommand{\om}{\omega}
\newcommand{\De}{\Delta}

\newcommand{\Om}{\Omega}

\newcommand{\Inf}{\mathit{Inf}}

\newcommand{\pfa}{\power A}

\newcommand{\ti}[1]{\tilde{#1}}
\newcommand{\ol}[1]{\bar{#1}}

\newcommand{\shDe}{\sh{\Delta}}

\newcommand{\bbW}{\mathbb{W}}

\newcommand{\ro}[1]{{#1}_I}
\newcommand{\Lan}{\mathrm{L}}
\newcommand{\Pow}{\mathcal{P}}
\newcommand{\coloneqq}{\mathrel{{:}{=}}}
\newcommand{\game}{\mathcal{G}}
\newcommand{\elo}{\exists}
\newcommand{\abe}{\forall}
\newcommand{\domain}{\mathsf{dom}}
\newcommand{\range}{\mathsf{rng}}
\newcommand{\probdist}{\mathsf{D}_\omega}
\newcommand{\multiset}{\mathsf{M}_\omega}
\newcommand{\KPF}{\mathrm{KPF}}

\newcommand{\nonempty}[1]{\game_{\not= \nada}(#1)}
\newcommand{\card}[1]{|#1|}

\usepackage[all]{xypic}

\def\doi{4 (4:10) 2008}
\lmcsheading%
{\doi}
{1--43}
{}
{}
{Jan.~27, 2007}
{Nov.~21, 2008}
{}   

\begin{document}

\title{Coalgebraic Automata Theory: Basic Results\rsuper*}

\author[C.~Kupke]{Clemens Kupke\rsuper a} 
\address{{\lsuper a}Imperial College London\\
180 Queen's Gate \\ London SW7 2AZ, United Kingdom}
\email{ckupke@doc.ic.ac.uk}
\thanks{{\lsuper a}Supported by the Netherlands Organisation
for Scientific Research (NWO) under FOCUS/BRICKS 
grant number 642.000.502.}

\author[Y.~Venema]{Yde Venema\rsuper b}
\address{{\lsuper b}Universiteit van Amsterdam,  ILLC \\ 
Plantage Muidergracht 24 \\ 1018 TV Amsterdam,
Netherlands} 
\email{Y.Venema@uva.nl}

\titlecomment{{\lsuper*}The paper is a completely revised and extended
version of~\cite{kupk:clos05}.}

\begin{abstract}
We generalize some of the central results in automata theory to the
abstraction level of coalgebras and thus lay out the foundations of
a universal theory of automata operating on infinite objects.

Let $\F$ be any set functor that preserves weak pullbacks. We show
that the class of recognizable languages of $\F$-coalgebras is closed
under taking unions, intersections, and projections. We also prove
that if a nondeterministic $\F$-automaton accepts some coalgebra it
accepts a finite one of the size of the automaton.
Our main technical result concerns an explicit construction which
transforms a given alternating $\F$-automaton into an equivalent
nondeterministic one, whose size is exponentially bound by the size
of the original automaton.
\end{abstract}

\maketitle


\section{Introduction}
\label{s:1}

\noindent An important branch of automata theory, itself one of the classical 
subdisciplines of computer science, concerns the study of finite automata 
as devices for classifying infinite, or possibly infinite, objects.
This perspective on finite automata has found important applications in 
areas of computer science where one investigates the ongoing behavior of 
nonterminating programs such as operating systems.
As an example we mention the automata-based verification method of
\emph{model checking}~\cite{clar:mode00}.
This research also has a long and strong theoretical 
tradition, in which an extensive 
body of knowledge has been developed, with a number of landmark results.
Many of these link the field to neighboring areas such as logic and game
theory, see~\cite{grae:auto02} for an overview.
The outstanding example here is of course Rabin's decidability theorem~%
\cite{rabi:deci69} for the monadic second order logic of trees; to mention
a more recent result, Janin \& Walukiewicz~\cite{jani:auto95} identified
the modal $\mu$-calculus as the bisimulation invariant fragment of the 
monadic second order logic of labelled transition systems.

An interesting phenomenon in automata theory is that most (but not all) key 
results hold for word and tree automata alike, and that many can even be 
formulated and proved for automata that operate on yet other objects such
as trees of unbounded branching degree, or labelled transition systems.
This applies for instance to various closure properties of the class of 
recognizable languages, and to the fact that alternating automata can be 
transformed into equivalent nondeterministic ones.
These observations naturally raise the question, whether these results can 
perhaps be formulated and proved at a more general level of abstraction.
This is by analogy to algebra, where it is more convenient, for instance,
to formulate and prove the Homomorphism Theorems at the level of universal 
algebra, than to treat them separately for every algebraic signature.
Of course, such a universal approach towards automata theory would first of 
all require the introduction of an abstract notion that generalizes structures 
like words, trees and transition systems.
Fortunately, such an abstract notion already exists in the form of 
\emph{coalgebra}.

The theory of universal coalgebra (see \cite{rutt:univ00} for an overview)
seeks to provide a general framework for the study of notions related to 
(possibly infinite) behavior, such as invariance and observational
indistinguishability (bisimilarity, in most cases).
Intuitively, coalgebras (as objects) are simple but fundamental mathematical 
structures that capture the essence of dynamics.
In this paper we will restrict our attention to \emdef{systems}; these are
state-based coalgebras consisting of a set $S$ and a map $S \to \F S$,
where $\F$ is some set functor determining the \emph{signature} of the
coalgebra.
The general theory of coalgebra has already developed some general tools for
the specification of properties of coalgebras.
In particular, starting with Moss' coalgebraic logic~\cite{moss:coal99}, 
several logical languages have been proposed, usually with a strong modal 
flavor.
Most of these languages are not designed for talking about \emph{ongoing}
behaviour, but in \cite{vene:auto04}, the second author introduced a 
coalgebraic fixed-point logic that does enable specifications of this kind
(see~\cite{vene:auto06} for a more detailed exposition).

The same paper~\cite{vene:auto04} also introduces, for coalgebras over a
standard set functor $\F$ that preserves weak pullbacks, the notion of an
\emph{$\F$-automaton} --- we will recall the definition in section~\ref{s:2}.
These automata provide a common generalization of the familiar automata that 
operate on specific coalgebras such as words, trees or graphs.
They also come in various shapes and kinds, the most important distinction
being between alternating, nondeterministic, and deterministic ones,
respectively.

Basically, $\F$-automata are meant to either accept or reject pointed 
coalgebras (that is, pairs $\struc{\bbS,s}$ 
consisting of an $\F$-coalgebra 
$\bbS$ together with a selected state $s$ in the carrier $S$ of $\bbS$),
and thus \emph{express properties of states} in $\F$-coalgebras.
This makes them very similar to formulas, and explains the close connection
with coalgebraic (fixed-point) logic.
This connection generalizes the relation between the modal $\mu$-calculus
and the $\mu$-automata~\cite{jani:auto95} to the abstraction level of 
coalgebra. 
A key feature of the coalgebraic framework is that one restricts attention
to observable, or \emph{bisimulation-invariant properties}.
Another important aspect of $\F$-automata involves game theory: the criterion 
under which an $\F$-automaton accepts or rejects a pointed coalgebra is 
formulated in terms of an infinite two-player graph game.

The aim of developing this coalgebraic framework is not so much to introduce
new ideas in automata theory, as to provide a common generalization for 
existing notions that are well known from the theory of more specific 
automata.
Apart from its general mathematical interest, this abstract approach may be 
motivated from various sources. 
To start with, the abstract perspective may be of help to find the right notion 
of automaton for other kinds of coalgebras, besides the well known kinds
like words and trees.
It may also be used to prove interesting results on coalgebraic logics ---
we will briefly come back to this in section~\ref{s:concl}.
\medskip

It is the aim of the present paper, which is a completely revised and extended
version of~\cite{kupk:clos05}, to provide further motivation for taking a 
coalgebraic perspective on automata, by showing that some of the key results
in automata theory can in fact be lifted to this more abstract level.
In particular, this allows for \emph{uniform} proofs of these results, which
in its turn may lead to a better understanding of automata theory as such.
The concrete results that we prove concern the relation between alternating 
and nondeterministic automata, some of the closure properties that one may
associate with automata, and the nonemptiness problem.
For a proper formulation, we need to develop some terminology.

A class of pointed $\F$-coalgebras will be referred to as an 
\emdef{$\F$-language}.
Such a language $\clL$ is \emdef{recognized by an $\F$-automaton $\bbA$} if
a pointed $\F$-coalgebra belongs to $\clL$ if and only if it is accepted by 
$\bbA$, and \emdef{(nondeterministically) recognizable} if it is
recognized by some (nondeterministic) parity $\F$-automaton.
Our main technical result can now be formulated as follows. 

\begin{fewtheorem}
\label{t:main}
Let $\F$ be some set functor that preserves weak pullbacks.\footnote{The 
meaning and importance of this side condition will be explained in
section~\ref{s:2}.}
Then every alternating parity $\F$-automaton $\bbA$ has a nondeterministic
equivalent $\bbA^{\bl}$.
Hence, an $\F$-language is recognizable iff it is nondeterministically
recognizable.
\end{fewtheorem}

More specifically, we will give a construction which is both \emph{uniform},
in the sense that it takes the type $\F$ of $\bbA$ as a parameter, and
\emph{concrete}, in that we give an explicit definition of $\bbA^{\bl}$ in
terms of $\bbA$.

In order to discuss closure properties, let $\Oop$ be some operation on 
$\F$-languages, then we say that a class of languages is closed under $\Oop$ 
if we obtain a language from this class whenever we apply $\Oop$ to a family 
of languages from the class.
For example, one may easily prove that recognizable $\F$-languages are closed
under taking intersection and union; with some more effort we will show that 
the class of nondeterministically recognizable $\F$-languages is closed 
under (existential) projection.
Theorem~\ref{t:main} allows us to strengthen the above list of closure
properties as follows.

\begin{fewtheorem}
\label{t:clos}
Let $\F$ be some set functor that preserves weak pullbacks.
Then the class of recognizable $\F$-languages is closed under union, 
projection and intersection.
\end{fewtheorem}

Conspicuously absent in this list is closure under complementation --- we 
will come back to this in section~\ref{s:concl}.

The third result that we want to mention here concerns the nonemptiness
problem for coalgebra automata.
As we will show, if a nondeterministic parity automaton $\bbA$ accepts
an $\F$-coalgebra at all, then it accepts an $\F$-coalgebra that somehow
`lives inside $\bbA$'.
From this and Theorem~\ref{t:main} the following result is immediate.

\begin{fewtheorem}
\label{t:nempt}
Let $\F$ be some set functor that preserves weak pullbacks.
Then the $\F$-language recognized by a parity automaton $\bbA$ is
nonempty iff $\bbA$ accepts a finite $\F$-coalgebra of bounded size.
\end{fewtheorem}

We \emdef{prove} these results by generalizing, to the coalgebraic level,
(well-)known ideas from the theory of specific automata.
This applies in particular to the results on graph automata by Janin \&
Walukiewicz~\cite{jani:auto95}, whereas our approach is similar to the
abstract universal algebraic approach of Niwi\'nski and Arnold%
~\cite{niwi:fixe97,arno:rudi01}.
Just as in the literature on specific kinds of automata, our proofs crucially
depend on the ``import'' of two fundamental results from the theory of automata
and infinite games. 
The proof of our main result, Theorem~\ref{t:main} contains an essential
instance of the determinization of $\omega$-automata (a result originally
due to MacNaughton~\cite{mcna:test66}, whereas the construction we use goes
back to Safra~\cite{safr:onth88}).
Theorem~\ref{t:nempt}, our solution to the nonemptiness problem for coalgebra
automata, can be seen as an application of the history-free determinacy of 
parity games (see Fact~\ref{f-2-3}).

Finally, let us stress again that our aim here is not to prove new results
for well-known structures.
Rather, the point that we try to make is that many of the well-known
theorems in automata theory in fact belong to the field of Universal
Coalgebra, in the same way that the Homomorphism Theorems are results 
in Universal Algebra.
In our opinion, the abstract, coalgebraic perspective not only generalizes
existing results, the uniformity of the coalgebraic presentation has helped
us to obtain a better understanding of automata theory itself.
To mention one example, for various automata-theoretic constructions, our
coalgebraic proofs show that generally, \emph{size issues} do not depend
on the type of the automata (that is, on the functor $\F$), because most of
these issues only depend on constructions for $\omega$-automata
(cf.\ also Remark~\ref{r:fund} below).

The paper is organized as follows: In Section~\ref{s:2} we equip the reader
with the necessary background material on coalgebra automata.
After that, we show in Section~\ref{s:altond} how to transform an alternating
parity $\F$-automaton into an equivalent nondeterministic $\F$-automaton
that has a so-called regular acceptance condition.
These `regular automata' will be discussed in detail in
Section~\ref{s:regaut},
where we prove that any regular $\F$-automaton can be transformed into 
an equivalent $\F$-automaton with a parity acceptance condition.
In Section~\ref{s:3} we combine the results from Section~\ref{s:altond}
and Section~4 in order to obtain Theorem~\ref{t:main}. 
We then discuss other closure properties of recognizable languages and
prove Theorem~\ref{t:clos}.
Finally, in Section~\ref{s:6}, we demonstrate that the nonemptiness problem
of $\F$-automata can be effectively solved, proving Theorem~\ref{t:nempt}.
Section~\ref{s:concl} concludes the paper with a short summary of our
results and an outlook on future research.

\paragraph{Acknowledgements}
We are grateful to the anonymous referees for providing a number of very 
helpful suggestions concerning the presentation of our results.

\section{Preliminaries}
\label{s:2}

\noindent Since coalgebra automata bring together notions from two
research areas (coalgebra and automata theory), we have made an effort
to make this paper accessible to both communities.  As a consequence,
we have included a fairly long section containing the background
material that is needed for understanding the main definitions and
proofs of the paper.  Readers who want to acquire more knowledge could
consult \textsc{Rutten}~\cite{rutt:univ00} and \textsc{Gr\"{a}del,
  Thomas \& Wilke}~\cite{grae:auto02} for further details on universal
coalgebra and automata theory, respectively.  For a more gentle
introduction to coalgebra automata, the reader is referred to
\textsc{Venema}~\cite{vene:auto06}.

First we fix some mathematical notation and terminology.

\begin{convention}{
Let $f:S \to T$ be a function. Then the \emph{graph} of $f$ is the relation
\[
\Graph{f} \coloneqq \{ (s,f(s)) \in S \times T \mid s \in S \}
\]
Given a relation $R \sse S \times T$, we denote the \emph{domain} and 
\emph{range} of $R$ by $\domain(R)$ and $\range(R)$, respectively.
Given subsets $S' \sse S$, $T' \sse T$, the \emph{restriction} of $R$ to  
$S'$ and $T'$ is given as
\[
R\rst{S' \times T'} \coloneqq R \cap (S' \times T').
\]
The composition of two relations $R \sse S \times T$ and $R' \sse T \times U$
is denoted by $R \circ R'$, whereas the composition of two functions $f: S 
\to T$ and $f': T \to U$ is denoted by $f'\circ f$.
Thus, we have $\Graph{f' \circ f} = \Graph{f}\circ\Graph{f'}$.
Finally, the set of all functions from $S$ to $T$ is denoted both by
$T^{S}$ and by $S \to T$.
}
\end{convention}

\subsection{Coalgebras, bisimulations and relation lifting}

Informally, a coalgebra $\bbS$ consists of a set $S$ of \emph{states},
together with a \emph{transition structure} $\si$ mapping a state $s \in S$
to an object $\si(s)$ in some set $\F S$, where $\F$ is the \emph{type} of 
the structure $\bbS$.
Technically, the type of a coalgebra is formalized as a \emph{set functor},
i.e. an endofunctor on the category $\Set$ that has sets as objects and 
functions as arrows.
Such a functor maps a set $S$ to a set $\F S$, and a function $f: S \to S'$
to a function $\F f: \F S \to \F S'$.

\begin{definition}
Given a set functor $\F$, an $\F$-coalgebra is a pair $\bbS =
\struc{S,\si}$ with $\si: S \to \F S$.
A \emdef{pointed coalgebra} is a pair consisting of a coalgebra together
with an element of (the carrier set of) that coalgebra.
\end{definition}

\begin{remark}
The action of $\F$ on functions is needed to define the \emph{morphisms}
between two coalgebras.
Formally, given two $\F$-coalgebras $\bbS = \struc{S,\si}$ and $\bbS' = 
\struc{S',\si'}$, a function $f: S \to S'$ is a coalgebra morphism if 
$\si'\circ f = (\F f) \circ \si$.
Since these morphisms do not play a large role in this paper, the reader
that has no familiarity with coalgebras may safely ignore this aspect.
In the sequel, we will often even introduce set functors just by defining
their action on sets, trusting that the initiated reader will be able to
supplement the action on functions.
\end{remark}

\begin{example}
\label{ex:BTK}\hfil
\begin{enumerate}[(1)]
\item
For instance, consider the functor $\BT$ which associates with a set $S$
the cartesian product $S \times S$, and with a map $f: S \to S'$, the map
$\BT f: S\times S \to S'\times S'$ given by $(\BT f)(s,s') := (f(s),f(s'))$.
Thus every state in an $\BT$-coalgebra has both a left and a right
successor.
As a special example of a $\BT$-coalgebra, consider the \emph{binary tree}
given as the set $2^{*}$ of finite words over the alphabet $2 =
\{0,1\}$, with the coalgebra map given by $s \mapsto (s0,s1)$.

\item
Directed graphs can be seen as coalgebras of the power set functor $\K$.
The functor $\K$ maps a set to its collection of subsets, and a function
$f: S \to S'$ to its direct image function $\K f: \power S \to \power S'$
given by $(\K f)(X) := \{ f(x) \in S' \mid x \in X \}$.
A graph $(G,E)$ is then modelled as the coalgebra $(G,\lambda x.E[x])$,
that is, the relation $E$ is given by the function mapping a point $x$ to 
the collection $E[x]$ of its (direct) successors.

\item 
The labelled transition systems from~\cite{jani:auto95} can be represented
as coalgebras for the functor $\Pow \Phi \times \K^{A}$, i.e.\ for the functor
that maps a set $S$ to the set $\Pow \Phi \times (\Pow S)^A$.
Here $\Phi$ is a set of propositional variables and $A$ is the set of actions.
Hence, in the coalgebraic presentation, a point $s$ is mapped to the pair
consisting of the set of proposition letters in $\Phi$ that are true at 
$s$, and, for each action $a$, the set of `$a$-successors' of $s$.

\item
For any set $C$, a $C$-coloring of a coalgebra $\bbS$ is a map $\ga: S \to 
C$; the $C$-colored $\F$-coalgebra  $\bbS \oplus \ga := \struc{S,\ga,\si}$
can be identified with the $\F_{C}$-coalgebra $\struc{S,\langle \ga,\si 
\rangle}$.
Here $\F_{C}$ is the functor that takes a set $S$ to the set $C \times S$
(and that takes a map $f: S \to S'$ to the function $\F_{C}f:  C \times S 
\to C \times S'$ given by $(\F_{C}f)(c,s) = (c,f(s))$).
So $\BT_{C}$ and $\K_{C}$ are, respectively, the functors that we may 
associate with $C$-labelled binary trees and $C$-labelled graphs,
respectively. 

\item Let
$\probdist$ be the functor that maps a set $S$ to the set
\[ \probdist S \coloneqq 
\big\{ \rho: S \to [0,1] \mid \rho
\; \mbox{has finite support and $\sum_{s\in S}$} 
\rho(s) = 1
\big\} 
\]
where we say that $\rho$ has finite support if  $\rho(s) \not= 0$
for only finitely many elements $s$ of $S$. Then 
coalgebras for the functor $1 + \probdist$
correspond to the probabilistic transition systems 
by Larsen and Skou in~\cite{lars:bisi91}. Here $1 + \probdist$ denotes
the functor that maps a set $S$ to the disjoint union of the one-element set
and the set $\probdist S$. Further details about this example can be
found in~\cite{rutt:bisi99}.

\item Transition systems in which every state
has a {\em multiset} of successors can be modeled as 
coalgebras for the functor
$\multiset$ that maps a set $S$ to the set $\multiset S$ of all
functions $\mu: S \to \bbN$ with finite support, i.e.\
$\mu(s) \not= 0$ for only finitely many elements $s$ of $S$.
More information can be found 
in~\cite{gumm:mono01}.
\end{enumerate}
\end{example}

\noindent As mentioned in the introduction, the theory of coalgebra
aims to provide a simple but general framework for formalizing and
studying the concept of \emph{behavior}.  For this purpose it is of
importance to develop a notion of \emph{behavioral equivalence}
between states.  For almost all important coalgebraic types, this
notion can be naturally expressed using \emph{bisimulations}.
Intuitively, these are relations between the state sets of two
coalgebras that witness the observational indistinguishability of the
pairs that they relate.  For our purposes it will be convenient to use
a definition of bisimilarity in terms of \emph{relation lifting}.

\begin{definition}
\label{d:rellift}
Let $\F$ be a set functor.
Given two sets $S$ and $S'$, and a binary relation $Z$ between $S \times S'$,
we define the \emdef{lifted relation} $\Fb (Z) \sse \F S \times \F
S'$
as follows:
\[
\Fb (Z) := \{ ((\F\pi) (\phi), (\F\pi')(\phi)) \mid \phi \in \F Z \},
\]
where $\pi: Z \to S$ and $\pi':Z \to S'$ are the projection functions given
by $\pi(s,s') = s$ and $\pi'(s,s') = s'$.
\[ \xymatrix{S & R \ar[l]_{\pi} \ar[r]^{\pi'} 
\save[]+<0cm,-0.3cm>
\ar@{=>}[d]+<0cm,0.3cm>^{\mbox{{\tiny lifting}}} \restore & S' \\
\F S & \F R \ar[r]_{\F \pi'} \ar[l]^{\F \pi} & \F S'}
\]

Now let $\bbS = \struc{S,\si}$ and $\bbS' = \struc{S',\si'}$ be two 
$\F$-coalgebras.
Then a relation $Z \sse S \times S'$ satisfying
\begin{equation}
\label{eq:bis}
(\si(s),\si'(s')) \in \Fb(Z)
\mbox{ for all } (s,s') \in Z.
\end{equation}
is an  \emph{($\F$-)bisimulation} between $\bbS$ and $\bbS'$.
Two states $s$ and $s'$ in such coalgebras are \emph{bisimilar}, notation:
$\bbS,s \bis \bbS',s'$, iff they are linked by some bisimulation.
\end{definition}

Intuitively, for $Z \sse S \times S'$ to be a bisimulation we require 
that whenever $s$ and $s'$ are linked by $Z$, then $\si(s)$ and $\si'(s)$
are linked by the lifted version $\Fb(Z)$ of $Z$.

\begin{remark}\label{r:well-def}
Strictly speaking, the definition of the relation lifting of a given relation 
$R$ depends on the type of the relation, i.e.\ given sets $S',S,T',T$ such 
that $R \sse S' \times T'$ and $R \sse S \times T$, it matters whether we 
look at $R$ as a relation from $S'$ to $T'$ or as a relation from $S$ to $T$. 
This possible source of ambiguity can be avoided if we require the functor 
$\F$ to be \emph{standard}.  
We come back to a detailed discussion of the notion of a standardness at the 
end of this section.
\end{remark}
 
\begin{example}
\label{ex:bis}
Let us see how these definitions apply to coalgebras for the functors 
$\BT_{C}$ and $\K$ of Example~\ref{ex:BTK}.
For this purpose, fix two sets $S$ and $S'$, and a relation $Z \sse S \times S'$.

For the definition of the relation $\ovl{\BT_{C}}(Z)$, it is easy to see that 
$(c,s_{0},s_{1})$ and $(c',s_{0}',s_{1}')$ are related iff $c = c'$ and
both $(s_{0},s_{0}')$ and $(s_{1},s_{1}')$ belong to $Z$.
Hence the relation $Z$ is a bisimulation between two $\BT_{C}$-coalgebras 
if $Z$-related points have the same color and both their left successors
and their right successors are $Z$-related.
From this it easily follows that two labelled binary trees are bisimilar iff
they are identical.
However, the notion becomes less trivial if we consider other coalgebras for
the functor $\BT_{C}$.
In fact, bisimilarity can be used to formulate some well-known notions in the
theory of tree automata; for instance, a labelled binary tree is 
\emph{regular} iff it is bisimilar to a finite $\BT_{C}$-coalgebra.

Concerning the Kripke functor $\K$, observe that $(X,X') \in \ovl{\K}(Z)$ iff
for all $x \in X$ there is an $x' \in X'$, and for all $x' \in X'$ there is
an $x \in X$ such that $(x,x') \in Z$.
That is, $\ovl{\K}(Z)$ is the \emph{Egli-Milner} lifting of $Z$.
Thus, a relation relating nodes of one graph $\bbS$ to those of another
graph $\bbS'$ is a \emph{bisimulation} if for every linked pair $(s,s')$,
every successor $t$ of $s$ is related to some successor $t'$ of $s'$, and,
vice versa, every successor $t'$ of $s'$ is related to some successor $t$ of
$s$.

Finally, the connection between $\Fb$ and $\ovl{\F_{C}}$ is given by the
following: $((c,\phi),(c',\phi')) \in \ovl{\F_{C}}(Z)$ iff $c=c'$ and 
$(\phi,\phi') \in \Fb(Z)$.
\end{example}

Given the key role that relation lifting plays in this paper, we need some
of its properties.
It can be shown that relation lifting interacts well with the operation of
taking the graph of a function $f: S \to S'$, and with most operations on
binary relations.
In fact, the properties listed in Fact~\ref{fact:wpb} are \emph{all} the 
information on relation lifting that is needed in this paper.
Readers that have no interest in categorical details may safely skip some
material and move on to Fact~\ref{fact:wpb}.

Unfortunately, for two properties of $\Fb$ which are crucial for our results,
we need the functor to satisfy a certain categorical property, namely the
preservation of weak pullbacks (to be defined in Remark~\ref{r:wpp} below.
Fortunately, many set functors in fact do preserve weak pullbacks, which 
guarantees a wide scope for the results in this paper.

\begin{example}
All functors from Example~\ref{ex:BTK} preserve weak pullbacks.
For the functors $\BT$ and $\K$ this is not difficult to check. 
A proof for the fact that $\probdist$ is weak pullback preserving can be
found in~\cite{rutt:bisi99}.
That $\multiset$ is weak pullback preserving is a direct consequence of the
results in~\cite{gumm:mono01}.
\end{example}

Furthermore it can be shown that the class of weak pullback preserving
functors contains all constant functors, and is closed under composition,
taking products, coproducts (disjoint unions), and under exponentiation. 
We use these observations in order to define a class of functors that are
all weak pullback preserving.

\begin{definition}\label{d:ekpf}
The class $\KPF$ of \emdef{Kripke polynomial functors} is inductively defined 
by putting
\[
\F \mathrel{::=} A \in \Set \mid \Id \mid \Pow \mid \F + \F \mid \F \times \F \mid 
\F^D, D \in \Set \mid \F \circ \F,
\]
where $\Id$ denotes the identity functor on $\Set$.
Enlarging the basis of this inductive definition with the probabilistic and
multiset functor of Example~\ref{ex:BTK}:
\[
\F \mathrel{::=} A \in \Set \mid \Id \mid \Pow \mid \probdist \mid \multiset 
\mid \F + \F \mid \F \times \F \mid \F^D, D \in \Set \mid \F \circ \F.
\]
we arrive at the definition of the \emdef{extended Kripke polynomial functors}.
\end{definition}

As explained in the example one can prove the following fact.

\begin{fact}
All extended Kripke polynomial functors preserve weak pullbacks.
\end{fact}

Although the precise definition of weak pullback preservation is not relevant
in order to understand this paper, for the interested reader we provide some
details in Remark~\ref{r:wpp} below.

\begin{remark}
\label{r:wpp}
Given two functions $f_{0}: S_{0} \to S$, $f_{1}: S_{1} \to S$, a \emph{weak
pullback} is a set $P$, together with two functions $p_{i}: P \to S_{i}$ such
that $f_{0}\circ p_{0} = f_{1} \circ p_{1}$, and in addition, for every
triple $(Q,q_{0},q_{1})$ also satisfying $f_{0}\circ q_{0} = f_{1}\circ
q_{1}$, there is an arrow $h: Q  \to P$ such that $q_{0} = h\circ p_{0}$ and
$q_{1} = h\circ p_{0}$.\\
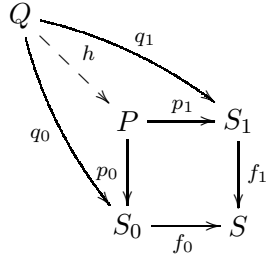
\centerline{
	 \xymatrix{ Q \ar@/_/[ddr]_{q_0}
      \ar@/^/[drr]^{q_1} \ar@{-->}[dr]^h & & \\
      & P \ar[r]^{p_1} \ar[d]_{p_0} & S_1 \ar[d]^{f_1}\\
      & S_0 \ar[r]_{f_0} & S }}

A functor $\F$ \emph{preserves weak pullbacks} if it transforms every weak
pullback $(P,p_{0},p_{1})$ for $f_{0}$ and $f_{1}$ into a weak pullback for
$\F f_{0}$ and $\F f_{1}$.
(The difference with \emph{pullbacks} is that in the definition of a weak
pullback, the arrow $h$ is not required to be unique.)

A category-theoretically nicer way of formulating this property
involves the category $\Rel$, i.e.\ the category of sets (as objects) and binary
relations (as arrows). 
A functor $\mathsf{Q}$ on $\Rel$ is called a \emdef{relator}
if for all binary relations $R,S$ such that $R \sse S$ 
we have $\mathsf{Q}(R)
\sse \mathsf{Q}(S)$.
A relator $\mathsf{Q}$
\emdef{extends} a functor $\F: \Set \to \Set$ if $\mathsf{Q} S = \F S$ for 
any object (set) $S$ and $\mathsf{Q} (\Graph{f}) = \Graph{\F f}$ for any
arrow (function) $f$; here $\Graph{f}$ denotes the graph of $f$.
Then one may prove that a set functor $\F$ can be extended to a relator 
iff $\F$ preserves weak pullbacks, 
%
%
and that this extension is unique if it exists.
For a sketch of the proof of this fact note that it is easy to see
(cf.~e.g.~\cite[Chap.~5]{bird:alge97}) that any relator $\mathsf{Q}$
extending a set functor $\F$ satisfies $\mathsf{Q}(R)=\Fb(R)$ for all
binary relations $R$. Furthermore Trnkov\'a
proved in~\cite{trnk:rela77} that $\Fb$ is a relator
iff $\F$ preserves weak pullbacks.

Finally, for an example of a functor that does not preserve weak pullbacks,
consider the functor that takes a set $S$ to the set of upward closed subsets
of $(\power S, \subseteq)$, and a function $f$ to its double inverse 
$(f^{-1})^{-1}$.
\end{remark}

As mentioned already, in the following fact we list all the properties
concerning relation lifting that we need in this paper.

\begin{fact}
\label{fact:wpb}
Let $\F$ be a set functor.
Then the relation lifting $\Fb$ satisfies the following properties, for
all functions $f: S \to S'$, all relations $R,Q \sse S \times S'$, and all 
subsets $T\sse S$, $T' \sse S'$:
\begin{enumerate}[(1)]
\item $\Fb$ extends $\F$:
$\Fb (\Graph{f}) = \Graph{\F f}$;

\item $\Fb$ preserves the diagonal:
$\Fb(\Id_{S}) = \Id_{\F S}$;

\item $\Fb$ commutes with relation converse:
$\Fb (\converse{R}) = \converse{(\Fb{R})}$;

\item $\Fb$ is monotone: 
if $R \subseteq Q$ then $\Fb (R) \subseteq \Fb (Q)$;

\item $\Fb$ distributes over composition: $\Fb(R \circ Q) = \Fb (R) \circ \Fb (Q)$,
if $\F$ preserves weak pullbacks.
\end{enumerate}
\end{fact}

For proofs we refer to \cite{moss:coal99,balt:logi00}, and references
therein.
The proof that Fact~\ref{fact:wpb}(5) depends on the property of weak
pullback preservation goes back to Trnkov\'a~\cite{trnk:rela77}.
In the remainder of the paper we will usually assume that all functors
that we consider preserve weak pullbacks (but we will always mention this 
explicitly when we formulate important results).

\subsection{Graph games}
\label{ssec-gg}

Two-player infinite graph games, or \emph{graph games} for short, are defined
as follows.
For a more comprehensive account of these games, the reader is referred 
to {\sc Gr\"adel, Thomas \& Wilke}~\cite{grae:auto02}.

First some preliminaries on sequences.
Given a set $A$, let $A^{*}$, $A^{\om}$ and $A^{\star}$ denote the 
collections of finite, infinite, and all, sequences over $A$, respectively.
(Thus, $A^{\star} = A^{*} \cup A^{\om}$.)
Given $\alpha \in A^{*}$ and $\beta \in A^{\star}$ we define the 
\emph{concatenation} of $\alpha$ and $\beta$ in the obvious way, and we 
denote this element of $A^{\star}$ simply by juxtaposition: $\alpha\beta$.
Given an infinite sequence $\alpha \in A^{\om}$, let $\Inf(\alpha)$ denote
the set of elements $a \in A$ that occur infinitely often in $\alpha$.

A graph game is played on a \emph{board} $B$, that is, a set of \emph{positions}.
Each position $b \in B$ \emph{belongs} to one of the two \emph{players}, $\eloi$
(\'Eloise) and $\abel$ (Ab\'elard).
Formally we write $B = B_{\eloi} \cup B_{\abel}$, and for each position $b$ we 
use $P(b)$ to denote the player $i$ such that $b \in B_{i}$.
Furthermore, the board is endowed with a binary relation $E$, so that each
position $b \in B$ comes with a set $E[b] \sse B$ of \emph{successors}.
Formally, we say that the \emph{arena} of the game consists of a directed
two-sorted graph $\bbB = (B_{\eloi} ,  B_{\abel}, E)$.

A \emph{match} or \emph{play} of the game consists of the two players moving
a pebble around the board, starting from some \emph{initial position} $b_{0}$.
When the pebble arrives at a position $b \in B$, it is player $P(b)$'s turn
to move; (s)he can move the pebble to a new position of their liking, but 
the choice is restricted to a successor of $b$.
Should $E[b]$ be empty then we say that player $P(b)$ \emph{got stuck} at
the position.
A \emph{match} or \emph{play} of the game thus constitutes a (finite or 
infinite) sequence of positions $b_{0}b_{1}b_{2}\ldots\ $ such that 
$b_{i}Eb_{i+1}$ (for each $i$ such that $b_{i}$ and $b_{i+1}$ are defined).
A \emph{full play} is either (i) an infinite play or (ii) a finite play in
which the last player got stuck.
A non-full play is called a \emph{partial} play.

The rules of the game associate a \emph{winner} and (thus) a \emph{looser} for
each full play of the game.
A finite full play is lost by the player who got stuck; the winning condition
for infinite games is given by a subset $\Ref$ of $B^{\omega}$ ($\Ref$ is short
for `referee'): our convention is that $\eloi$ is the winner of $\beta \in
B^{\om}$ precisely if $\beta \in \Ref$.
A \emph{graph game} is thus formally defined as a structure $\G = (B_{\eloi}
,  B_{\abel}, E, \Ref)$.
Sometimes we want to restrict our attention to matches of a game with a
certain initial position; in this case we will speak of a game that is 
\emph{initialized} at this position.

Various kinds of winning conditions are known.
In a \emph{parity game}, the set $\Ref$ is defined in terms of a \emph{parity
function} on the board $B$, that is, a map $\Om: B \to \om$ with finite
range.
More specifically, the set $\Ref$ is defined by
\begin{equation}
\label{e-p-1}
B^{\om}_{\Om} := \{ \beta \in B^{\om} \mid
   \max \big( \Inf(\Om\circ\beta) \big) \mbox{ is even} \}
\end{equation}
(where $\Inf$ was defined at the beginning of this subsection).
In words, $\eloi$ wins a match if the highest parity encountered infinitely
often during the match, is even.

A \emph{strategy} for player $i$ is a function mapping partial plays $\beta
= b_{0}\cdots b_{n}$ with $P(b_{n}) = i$ to admissible next positions, that
is, to elements of $E[b_{n}]$.
In such a way, a strategy tells $i$ how to play:
a play $\beta$ is \emph{conform} or \emph{consistent with} strategy $f$ for 
$i$ if for every proper initial sequence $b_{0}\cdots b_{n}$ of $\beta$ with
$P(b_{n}) = i$, we have that $b_{n+1} = f(b_{0}\cdots b_{n})$.
A strategy is \emph{history free} if it only depends on the current 
position of the match, that is, $f(\beta) = f(\beta')$ whenever $\beta$ and
$\beta'$ are partial plays with the same last element (which belongs to the
appropriate player).
Occasionally, it will be convenient to extend the name `strategy' to
arbitrary functions mapping partial plays to positions; in order words, we
allow strategies enforcing \emph{illegal} moves.
In this context, the strategies proper, that is, the ones that dictate
admissible moves only will be called \emph{legitimate}.

A strategy is \emph{winning for player $i$} from position $b \in B$ if it
guarantees $i$ to win any match with initial position $b$, no matter how
the adversary plays --- note that this definition also applies to positions
$b$ for which $P(b) \neq i$.
A position $b \in B$ is called a \emph{winning position} for player $i$, if 
$i$ has a winning strategy from position $b$; the set of winning positions
for $i$ in a game $\G$ is denoted as $\Win_{i}(\G)$.

Parity games form an important game model because they have many attractive
properties, such as \emph{history-free determinacy}.

\begin{fact}
\label{f-2-3}
Let $\G = (B_{\eloi},B_{\abel},E,\Om)$ be a parity graph game. 
Then
\begin{enumerate}[(1)]
\item $\G$ is determined: $B = \Win_{\eloi}(\G) \cup \Win_{\abel}(\G)$.
\item Each player $i$ has a history-free strategy which is winning from 
any position in $\Win_{i}(\G)$.
\end{enumerate}
\end{fact}

The determinacy of parity games follows from a far more general
game-theoretic result concerning Borel games, due to 
\textsc{Martin}~\cite{mart:bore75}.
The fact that winning strategies in parity games can always be taken to be 
history free, was independently proved in 
\textsc{Mostowski}~\cite{most:game91}
and \textsc{Emerson \& Jutla}~\cite{emer:tree91}.

\subsection{Coalgebra automata and the acceptance game}

For a detailed exposition of coalgebra automata, the reader is referred
to \textsc{Venema}~\cite{vene:auto06}.
Here we confine ourselves to a self-contained survey of the definitions.

Probably the easiest introduction to coalgebra automata involves a
reformulation of the notion of bisimilarity in game-theoretic terms.
It follows from the characterization
(\ref{eq:bis}) that a bisimulation between two coalgebras $\bbS =
\struc{S,\si}$ and $\bbA = \struc{A,\al}$ is nothing but a postfixpoint of
the following operation on the set $\Sb{S \times A}$ of binary relations
between $S$ and $A$:
\[
Z \mapsto \{ (s,a) \in S \times A \mid 
(\si(s),\al(a)) \in \Fb(Z) \}.
\]
By monotonicity of relation lifting, this operation is monotone, and thus
it is an immediate consequence of standard fixpoint theory that there is a
largest bisimulation between the two coalgebras, which is given as the union
of all bisimulations between $\bbS$ and $\bbA$.
Furthermore, this largest bisimulation has a nice \emph{game-theoretic}
characterization, formulated in this generality for the first time in
\textsc{Baltag}~\cite{balt:logi00}.

For an informal description of this game, the admissible moves of the players
 are given as follows:
\begin{enumerate}[$\bullet$]
\item
in position $(s,a)$, $\eloi$ may choose a \emph{local bisimulation} for 
$s$ and $a$, i.e., a relation $Z \sse S \times A$ satisfying
$(\si(s),\al(a)) \in \Fb Z$;
\item
in position $Z \sse S \times A$, $\abel$ may choose any element $(s',a')$ 
of $Z$.
\end{enumerate}
Finally, the winning conditions for infinite matches of this game are 
straightforward: 
if $\eloi$ manages to survive all finite stages of a match,
she is declared the winner of the resulting infinite match.

This bisimilarity game can be formulated as a graph game with a very simple
parity winning condition 
(namely, all positions have the same, even,
priority).

\begin{definition}\label{def:bisgame}
Let $\F$ be a set functor, and let $\bbS = \struc{S,\si}$ and 
$\bbA = \struc{A,\al}$
be two $\F$-coalgebras.
The \emdef{bisimilarity game} $\B(\bbA,\bbS)$ associated with $\bbS$ and
$\bbA$ is the parity graph game $(B_{\eloi},B_{\abel},E,\Om)$ with
\[\begin{array}{lll}
   B_{\eloi} &:=& S \times A
\\ B_{\abel} &:=& \Sb{S \times A},
\end{array}\]
while $E$ and $\Om$ are given in Table~\ref{tb:1}.
Positions of the form $(s,a) \in S \times A$ are called \emdef{basic}.
\begin{table*}[t]
\begin{center}{\small 
\begin{tabular}[b]{|l|c|l|c|}
\hline
Position: $b$ & $P(b)$ & Admissible moves: $E[b]$ & $\Om(b)$ \\ \hline
$(s,a) \in  S \times A$           & $\eloi$ &
   $\{ Z \in \Sb{S \times A} \mid (\si(s),\al(a)) \in \Fb Z \}$    & 0 \\
$Z \in \Sb{S \times A}$           & $\abel$ & $Z$             & 0 \\ \hline
\end{tabular}
}\end{center}
\caption{Bisimilarity game for $\F$-coalgebras}
\label{tb:1}
\end{table*}
The set of winning positions for $\eloi$ in this game is denoted as
$\Win_{\eloi}(\B(\bbA,\bbS))$, or $\Win_{\eloi}$ if no confusion is likely.
\end{definition} 

\begin{remark}
We leave it for the reader to verify that
\[
(s,a) \in \Win_{\eloi}(\B) \mbox{ iff } \bbS,s \bis \bbA,a.
\]
The key observation here is that the relation 
$\{ (s,a) \in S \times A \mid (s,a) \in \Win_{\eloi}(\B)\}$ is the largest
bisimulation between $\bbS$ and $\bbA$.
\end{remark}

In order to make the key observation in understanding coalgebra automata,
we take a slightly \emph{different perspective} on the bisimilarity game.
The point is to think of one structure ($\bbA$) as \emph{classifying} the
other one ($\bbS$).
This conceptual removal of the symmetry between the two structures,
enables us to think of $\bbA$ as an \emph{automaton} operating on $\bbS$,
and of the game as an \emph{acceptance game} rather than as a comparison
game.
But from this perspective it is natural to impose some modifications on 
$\bbA$, making it resemble the standard concept of an automaton more
closely:

To start with, the state space $A$ of $\bbA$ is required to be finite,
and $\bbA$ will have a fixed \emph{initial state} $\ai \in A$.

Second, some infinite matches may be won by $\abel$.
This can easily be implemented as follows.
With each infinite match of the game we may associate an infinite stream of
\emph{basic positions} of the form $(s,a) \in S \times A$, and thus an
infinite stream of states from $A$. 
Hence, the winning conditions in the game can be formulated using a subset
$\Acc \sse A^{\om}$.
And since $A$ is finite we may formulate more specific conditions; for 
instance, a parity condition can be formulated by a map $\Om: A \to \om$.

And finally, we may introduce \emdef{nondeterminism} or even \emph{alternation}
on $\bbA$.
Here the idea is that, whereas in a coalgebra $\bbA = \struc{A,\al}$,
the `successor object' of a state $a \in A$ is fixed as $\al(a) \in \F A$,
we may now allow $\eloi$ to dynamically \emph{choose} this object from some 
\emph{set} $\De(a) \in \power(\F A)$ of objects in $\F A$.
Or, the `successor object' of $a \in A$ may be dynamically determined via
some game-theoretic interaction between the \emph{two} players.

Putting these observations together, we arrive at the following definition
of \textsc{Venema}~\cite{vene:auto04}.

\begin{convention}{\rm
In the sequel we will frequently denote the power set $\power S$ of a 
set $S$ by either $\pwE S$ or $\pwA S$.
This notation indicates that we are in a game-theoretic context, where $X
\in \power_{i} S$ means that $X$ represents a collection of possible moves,
and that $i$ is the player who may choose an element from $X$.
}\end{convention}

\begin{definition}\label{def:Fautomaton}
Let $\F$ be some set functor.
An (alternating) \emdef{$\F$-automaton} is a quadruple $\bbA = 
\struc{A,\ai,\Delta,\Acc}$ with $A$ some finite set of objects called
\emdef{states}, $\ai \in A$ the \emdef{initial state}, $\Delta: A \to \pwE \pwA
\F A$ the \emdef{transition function} and $\Acc \sse A^{\om}$ the 
\emdef{acceptance condition}.
An $\F$-automaton is called \emdef{nondeterministic} if each member of each 
$\Delta(a)$ is a singleton set.
The \emdef{size} of $\bbA$ is defined as the number of elements of $A$.

A \emdef{parity} $\F$-automaton is an $\F$-automaton $\bbA = 
\struc{A,\ai,\Delta,\Acc}$ where $\Acc$ is given by some 
parity condition $\Om: A \to \omega$;
such a structure will usually be denoted as $\bbA = \struc{A,\ai,\Delta,\Om}$.
The \emdef{index} of a parity automaton is defined as the size of the range of
$\Om$.\qquad\phantom{x}
\end{definition}

$\F$-automata are designed to accept or reject pointed $\F$-coalgebras.
The acceptance condition is formulated in terms of a graph game.
For an informal description of this game, the first observation is that 
matches of this game proceed in \emph{rounds} that start and end in a 
\emph{basic position} of the form $(s,a) \in S \times A$.
From such a basic position $(s,a)$, a round of the match proceeds along
the moves (a) - (d) below (of course, unless one of the players gets stuck):
\begin{enumerate}[(a)]
\item
$\eloi$ picks an element $\Phi \in \De(a)$, making $(s,\Phi) \in S
\times \pwA \F A$ the next position.
\item
$\abel$ picks $\phi \in \Phi$, moving to position $(s,\phi) \in S
\times \F A$.
\end{enumerate}
Note that this interaction between $\eloi$ and $\abel$ has fixed the 
`successor object' $\phi \in \F A$ of $a$, whereas the `successor object' 
$\si(s)$ of $s$ was determined from the outset of the match.
\begin{enumerate}[(a)]
\item[(c)]
$\eloi$ picks a `local bisimulation' $Y$ for $\phi$ and $\si(s)$, that is, 
a binary relation $Y \sse S \times A$ such that $(\si(s),\phi) \in \Fb Y$.
This relation $Y_{s,\phi}$ is itself the new position.
\item[(d)]
$\abel$ chooses a pair $(t,b) \in Y$ as the next basic position.
\end{enumerate}
In the sequel we will refer to the first two moves in the round of the game
as the \emph{static} part of the round (static because the match does not
pass to another state in the coalgebra), and to the last two moves as the 
\emph{dynamic} or \emph{coalgebraic} stage of the round.

For the winning conditions, recall that finite matches are lost by the
player who gets stuck. 
For infinite matches, consider an arbitrary such match:
\[
\mu \;=\; 
(s_{0},a_{0})(s_{0},\Phi_{0})(s_{0},\phi_{0})Y_{0}
(s_{1},a_{1})(s_{1},\Phi_{1})(s_{1},\phi_{1})Y_{1}
(s_{2},a_{2})\ldots
\]
Clearly, $\mu$ induces an infinite sequence of basic positions
\[
(s_{0},a_{0})(s_{1},a_{1})(s_{2},a_{2})\ldots
\]
and, thus, an infinite sequence of states in $A$:
\[
\mu\bpr_{A}:= a_{0}a_{1}a_{2}\ldots
\]
Now the winner of the match is determined by whether $\mu\bpr_{A}$ belongs
to the set $\Acc$ or not.

\begin{definition}\label{def:accgame}
Let $\bbA = \struc{A,\ai,\De,\Acc}$ 
be an $\F$-automaton, and let $\bbS = \struc{S,\si}$ be
an $\F$-coalgebra.
The \emdef{acceptance game} $\G(\bbA,\bbS)$ associated with $\bbA$ and $\bbS$
is the graph game $(B_{\eloi},B_{\abel},E,\udl{\Acc})$ with
\[\begin{array}{lllll}
   B_{\eloi} &:=& S \times A           &\cup& S \times \F A
\\ B_{\abel} &:=& S \times \power \F A &\cup& \Sb{S \times A},
\end{array}\]
where $E$ is given in Table~\ref{tb:2}.
Positions of the form $(s,a) \in S \times A$ are called \emdef{basic}.
\begin{table*}[t]
\begin{center}{\small 
\begin{tabular}[b]{|l|c|l|}
\hline
Position: $b$ & $P(b)$ & Admissible moves: $E[b]$  \\ \hline
$(s,a) \in  S \times A$           & $\eloi$ &
   $\{ (s,\Phi) \in S \times \Sb{\F A} \mid \Phi \in \De(a) \}$  \\
$(s,\Phi) \in S \times \Sb{\F A} $  & $\abel$ &
   $\{ (s,\phi) \in S \times \F A \mid
       \phi \in \Phi \}$               \\
$(s,\phi) \in S \times \F A$ & $\eloi$ &
   $\{ Z \in \Sb{S \times A} \mid (\si(s),\phi) \in \Fb Z \}$    \\
$Z \in \Sb{S \times A}$           & $\abel$ & $Z$              \\ \hline
\end{tabular}
}\end{center}
\caption{Acceptance game for an $\F$-automaton}
\label{tb:2}
\end{table*}

For the winning conditions of $\G(\bbA,\bbS)$, observe that every infinite
match $\mu$ induces an infinite sequence $\mu\bpr_{A} := a_{0}a_{1}a_{2}\ldots
\in A^{\om}$.
We put $\mu \in \udl{\Acc}$ if $\mu\rst{A} \in \Acc$, i.e.\
the winner of $\mu$ is $\eloi$ if $\mu\rst{A} \in \Acc$, and $\abel$ 
otherwise.

The set of winning positions for $\eloi$ in this game is denoted as
$\Win_{\eloi}(\G(\bbA,\bbS))$, or $\Win_{\eloi}$ if no confusion is likely.
$\bbA$ \emph{accepts} the pointed $\F$-coalgebra $(\bbS,s)$ if $(s,\ai) \in 
\Win_{\eloi}$.
\end{definition} 

\begin{remark}
It is clear from the definition of $\udl{\Acc}$ that only the \emdef{basic} 
positions of a match, i.e., positions of the form $(s,a) \in S \times A$, are 
relevant to determine the winner of the match.
Accordingly, in the sequel we will frequently represent a match of the game 
by the sequence of basic positions visited during the match.
\end{remark}

It is easy to see that the acceptance games associated with \emph{parity} 
automata, are parity games. 
(Simply define the priority of a basic position $(s,a)$ as $\udl{\Om}(s,a)
:= \Om(a)$, putting $\udl{\Om}(p) := 0$ for all other positions.)
But parity games are known to enjoy a strong form of determinacy: in any
position of the game board either $\eloi$ or $\abel$ has a history-free
winning strategy.
Therefore we can focus on $\eloi$'s history-free strategies.

\begin{definition}\label{def:eloistrategy}
Given an $\F$-coalgebra $(S,\sigma)$ and a parity $\F$-automaton $\bbA$ a 
\emdef{positional} or \emdef{history-free} strategy for $\eloi$ is a pair of
functions 
\[
(\Phi:S \times A \to \pwA \F A, \; Z:S \times \F A \to \Sb{S \times A}).
\]
Such a strategy is \emdef{legitimate} at a position if it maps the position 
to an admissible next position.
A positional strategy of the kind $\Phi: S \times A \to \power \F A$ will often
be represented as a map $\Phi: S \to (\power \F)^{A}$; values of this map 
will be denoted as $\Phi_s$, etc.
\end{definition}

\begin{remark}
\label{r:nd}
In the case of a nondeterministic $\F$-automaton $\bbA$ we shall usually
simplify our notation a little.
Recall that the transition map of such an automaton is of the form $\De:
A \to \pwE \pwA \F A$, with each element of each $\De(a)$ a singleton.
As a consequence, the move of $\forall$ in the static part of the game is
completely determined --- he has nothing to choose.
Consequently, we may eliminate these vacuous moves from the game by
simplifying the presentation of the automaton.

Identifying singleton sets with their unique elements, we think of the
transition function $\De$ as a map of type $A \to \pwE \F A$.
Accordingly then, we present the first component $\Phi$ of a positional
strategy $(\Phi,Y)$ for $\eloi$ as a function of type $S \times A \to \F A$.
\end{remark}

It should be stressed that, in the case of automata operating on well-known
infinite objects such as labelled binary trees, we have not really introduced
a new kind of device, but rather, given a slightly different presentation of
the more standard automata.

\begin{example}
Consider the case of binary tree automata over an alphabet $C$.
In our presentation, $C$-labelled binary trees are coalgebras for the functor
$\BT_{C}$, with $\BT_{C}(S) = C \times S \times S$, see Example~\ref{ex:BTK}. 

The transition map of nondeterministic tree automata is usually presented in 
the form
\[
\De: A\times C \to \power(A \times A),
\]
whereas in our presentation, following Remark~\ref{r:nd}, the transition map is of the form
\[
\De: A \to \pwE(C \times A \times A).
\]
It is not difficult to see that these two presentations are in fact
equivalent.
Using the principle of currying $(P \times Q) \to R \cong P \to (Q \to R)$,
and the notion of characteristic function $\power(Q) \cong Q \to 2$, we 
obtain
\begin{eqnarray*}
(A \times C) \to \power(A \times A)
   &\cong& (A \times C) \to ((A \times A) \to 2)
\\ &\cong& (A \times C \times A \times A) \to 2
\\ &\cong& A \to ((C \times A \times A) \to 2)
\\ &\cong& A \to \power(C \times A \times A).
\end{eqnarray*}
In~\cite{vene:auto06} the second author explains the equivalence between the 
two presentations in detail.
\end{example}

In fact, in the case the functor is of the form $\F_{C}$ for some functor
$\F$ and color set $C$ --- that is, if we are investigating $C$-colored
$\F$-coalgebras, we could have defined coalgebra automata in a different way,
which is more in line with the standard usage.
This alternative definition leads to the notion of
\emph{chromatic $\F$-automata}~\cite{vene:auto04}, which is needed (only) in
subsection~\ref{ss:clpr}.

\begin{definition}\label{def:chromatic}
Let $C$ be a finite set.
A \emdef{$C$-chromatic $\F$-automaton} is a quintuple $\bbA = 
\struc{A,\ai,C,\De,\Om}$ 
such that $\De: A \times C \to \pwE\pwA \F A$ (and $A$, $\ai$, and $\Om$ 
are as before).
Given such an automaton and a $\F_{C}$-coalgebra $\bbS = (S,\ga,\si)$, 
the acceptance game $\G_{C}(\bbA,\bbS)$ is defined as the acceptance game
for $\F$-automata with the only difference that $\eloi$ has to
move from a position $(s,a)$ to a position $(s,\Phi)$ such
that $\Phi \in \Delta(a,\gamma(s))$.
\phantom{bla}
\end{definition}

It was shown in~\cite{vene:auto04} that $C$-chromatic $\F$-automata and 
$\F_{C}$-automata have the same recognizing power.
We need the following fact.
\begin{fact}
\label{f:2:1}
With any parity $\F_{C}$-automaton $\bbA$ we may associate a $C$-chromatic 
$\F$-automaton $\bbA_C$, the {\em $C$-chromatic $\F$-companion} of $\bbA$, such 
that $\bbA$ and $\bbA_C$ accept the same $\F_{C}$-coalgebras. 
\end{fact}

\subsection{Standardization}
\label{s:standard}

We already mentioned that we will work with a weak pullback preserving
functor $\F$ throughout the paper. Moreover, for a smooth presentation,
it will sometimes be useful to require the functor $\F$ to be \emph{standard}.
Of course we will always clearly state when exactly we assume the property of 
standardness. 
The purpose of this section is to convince the reader that the restriction
to standard set functors is not essential, in that every set functor is 
`almost standard'. 
Let us start by formally defining the notion of a standard set functor.

\begin{definition}\label{d:stand}
Given two sets $S$ and $T$ such that $S \sse T$, let $\iota_{S,T}$ denote the inclusion map from
$S$ into $T$.
A set functor $\F$ is \emdef{standard} if $\F\iota_{S,T} = \iota_{\F S,\F T}$
for every inclusion map $\iota_{S,T}$.
\end{definition}

Many but not all set functors have this property.
For instance, all Kripke polynomial functors of Definition~\ref{d:ekpf} are
standard, but not the multiset functor of Example~\ref{ex:BTK}.

In words, a set functor is standard iff it turns inclusions into inclusions.
This means that in particular, $S \sse T$ implies $\F S \sse \F T$.
%
An immediate observation is that standardness ensures that 
the definition of 
the lifting of a relation $R$ is independent of its type (cf.\ 
Remark~\ref{r:well-def}).
\begin{proposition}
  Let $\F$ be a standard set functor, let $S',S,T',T$ be sets 
  and let $R' \subseteq S' \times T'$ be a relation. Furthermore
  let $R \sse S \times T$ be the relation $R'$ - but now seen
  as a relation between $S$ and $T$. Then the relations
  $\Fb R$ and $\Fb R'$ are equal.
\end{proposition}
\begin{proof} 
In order to prove the proposition, let $R' \subseteq S' \times T'$
and $R \sse S \times T$ represent the same relation, but with different
type information. We prove the proposition under the additional assumption 
that $S' \sse S$ and $T'\sse T$. 
The case in which this is not true can be reduced to this special case by 
considering the relation $R'' \coloneqq R \cap ( (S'\cap S) \times (T' \cap T))$, 
which has to be equal to both $R$ and $R'$. 
It then follows from our simpler claim that $\Fb R = \Fb R''= \Fb R'$.
  
  We now turn to the proof for the case that $S' \sse S$ and $T'\sse T$.
  The situation can be summarized by the following commuting diagram:
  \[ \xymatrix{
    S' \ar[d]_{\iota_{S',S}} & R' \ar[l]_{\pi_1'} \ar[r]^{\pi_2'}
    \ar@{=}[d]& 
    T' \ar[d]^{\iota_{T',T}} \\
    S &\ar[l]^{\pi_1} R \ar[r]_{\pi_2} & T}
  \]
  Here the $\iota$'s denote the inclusion maps. Moreover one has to 
  keep in mind that
  $R'$ and $R$ denote the same set - we only use two distinct letters in order
  to be able to distinguish between the two representations
  of the relation $R$. If we apply the standard
  functor $\F$ to this diagram
  we get
   \[ \xymatrix{
    \F S' \ar[d]_{\F\iota_{S',S}=\iota_{\F S',\F S}} 
    & \F R' \ar[l]_{\F \pi_1'} \ar[r]^{\F \pi_2'}
    \ar@{=}[d]& 
    \F T' \ar[d]^{\F \iota_{T',T}= \iota_{\F T',\F T}} \\
    \F S &\ar[l]^{\F \pi_1} \F R \ar[r]_{\F\pi_2} & \F T}
  \] 
Therefore we can calculate that
\begin{eqnarray*}
   (x,y) \in \Fb R' & \mbox{iff} & 
\mbox{there is a } z \in \F R' 
    \mbox{ with } \F \pi_1'(z) = x \mbox{ and } \F \pi_2'(z) = y 
\\ & \mbox{iff} & \mbox{there is a } z \in \F R' 
    \mbox{ with } (\iota_{\F S',\F S} \circ  \F \pi'_1) (z) = \F \pi'_1(z) = x
\\ & & \mbox{ and } (\iota_{\F T',\F T} \circ  \F \pi'_2) (z) = \F \pi'_2(z) = y
\\ & \mbox{iff} & \mbox{there is a } z \in \F R 
    \mbox{ with } \F \pi_1(z) = x \mbox{ and } \F \pi_2(z) = y 
\\ & \mbox{iff} & (x,y) \in \Fb R
  \end{eqnarray*}
\end{proof}

As already mentioned every weak pullback preserving set functor is
'almost' standard. This statement is made formal using the notion
of a natural isomorphism between functors.

\newcommand{\stl}[1]{{#1}_{\la}}
\newcommand{\Fl}{\stl{\F}}
\newcommand{\Flb}{\stl{\Fb}}

\begin{definition}
Let $\F$ and $\FG$ be two set functors, and suppose that for every set $S$
there is a bijection $\la_{S}: \F S \to\FG S$.
This collection $\la$ is a \emdef{natural isomorphism} between $\F$ and $\FG$,
if $(\FG f) \circ \la_{S} = \la_{T}\circ (\F f)$,
for every $f: S \to T$:
\[ 
\xymatrix{ S \ar[d]_f & & \F S \ar[r]^{\la_S}
    \ar[d]_{\F f} & \FG S \ar[d]^{\FG f} \\
    T & & \F T \ar[r]_{\la_T} & \FG T} 
\]
In this situation, we say that $\F$ is \emdef{naturally isomorphic} to $\G$
via  $\la$, notation: $\la: \F \cong \FG$.
If $\F$ is naturally isomorphic to a standard functor $\FG$, we call $\FG$ 
a \emdef{standardization} of $\F$.
\end{definition}

Naturally isomorphic functors are `almost the same'.
For instance, it is not hard to show that if $\F$ and $\FG$ are naturally
isomorphic, then the categories of $\F$-coalgebras is isomorphic to that of
$\FG$-coalgebras.
The following fact shows that the requirement of the functor $\F$ to
be standard is
not essential at all.
A proof of this fact can be based on the construction in part (a) of the
proof of Theorem~III.4.5 in \cite{adam:auto90}.\footnote{Note that the
construction in {\em loc.cit.\ } requires the functor to preserve arbitrary 
monomorphisms. 
It is not difficult to see that weak pullback preserving functors meet this
requirement.} 
%

\begin{fact}
\label{f:wppst}
Every weak pullback preserving functor has a standardization.
\end{fact}

\section{From alternation to nondeterminism}
\label{s:altond}

\noindent In this section we construct, for an arbitrary alternating parity
$\F$-automaton, an equivalent \emdef{nondeterministic} automaton.
Throughout this section we will be working with a fixed (but arbitrary)
$\F$-automaton $\bbA = \struc{A,\ai,\De,\Omega}$.
Before going into the technical details of the construction, let us briefly
mention the intuitions behind our approach.

\begin{remark}
These intuitions ultimately go back to ideas of Muller and Schupp, see for 
instance~\cite{mull:simu95}, but in particular, our proof generalizes work
by Janin and Walukiewicz~\cite{jani:auto95}, using the approach of
Arnold and Niwi{\'n}ski~\cite{arno:rudi01}.
In fact, with some effort, it would be possible to prove our result here
as a \emph{corollary} of the work mentioned.
That, however, would be to miss our point that a \emph{coalgebraic} proof is
possible, which is both \emph{uniform}, in the sense that it is parametrical
in the functor $\F$, and \emph{concrete} in the sense that we give an 
explicit definition which constructs the nondeterministic equivalent.
\end{remark}

Consider a single round of the acceptance game $\G(\bbA,\bbS)$ for some
$\F$-coalgebra $\bbS$, starting at a basic position $(s,a) \in S \times A$,
with $\eloi$ employing some positional strategy $(\Phi,Z)$:
\begin{enumerate}[$\bullet$]
\item 
$\eloi$ picks $\Phi_{s,a} \in \De(a)$, moving to position $(s,\Phi_{s,a})$; 
\item
$\abel$ picks $\phi \in \Phi$, moving to position $(s,\phi)$; 
\item
$\eloi$ picks $Y_{s,\phi} \sse S \times A$ with $(\si(s),\phi) \in \Fb 
Y_{s,\phi}$
--- this $Y_{s,\phi}$ is the new position;
\item
$\abel$ picks $(t,b) \in Y_{s,\phi}$ as the next basic position.
\end{enumerate}

\noindent
Our proof is based on the following four ideas:

\begin{enumerate}[\hbox to6 pt{\hfill}]
\item\noindent{\hskip-10 pt\bf strategic normal form:}\
First, we may bring the players' interaction pattern $\eloi\abel\eloi\abel$
in each round of the acceptance game for $\bbA$, into the \emph{strategic 
form} $\eloi\abel$ (or more precisely: $\eloi\eloi\abel\abel$).
Concretely, instead of choosing a separate $Y_{s,\phi}$ for each $\phi\in
\Phi$,  we will show that $\eloi$ may in fact choose the same relation
$\bigcup_{\phi\in\Phi}Y_{s,\phi}$ in response to each $\phi \in
\Phi$ picked by $\abel$.

\item\noindent{\hskip-10 pt\bf relations as states:}\
The crucial part of our proof involves a natural refinement of the classical
power set construction which is used for the determinization of automata
operating on finite words.
We will define a nondeterministic automaton $\sh{\bbA}$ based on the set
$\relA : = \Pow(A \times A)$ 
of \emph{binary relations} over the state space $A$
of $\bbA$.
The acceptance condition of $\sh{\bbA}$ is phrased in terms of \emph{traces}
through infinite sequences of such relations.

\item\noindent{\hskip-10 pt\bf regular automata:}\
The nondeterministic automaton $\sh{\bbA}$ is nonstandard in the sense that
its acceptance condition is expressed as an $\om$-regular language $\Acc$ over
the set $A$.
In the next section we will show that any such automaton is equivalent to a
standard nondeterministic automaton which is obtained as a kind of wreath
product of $\sh{\bbA}$ with the deterministic word automaton recognizing the
set of infinite $A$-words in $\Acc$.

\item\noindent{\hskip-10 pt\bf coalgebraic perspective:}\
While none of the above ideas in itself is essentially new, we believe that
our coalgebraic perspective simplifies matters.
It enables us to carry out the entire construction uniformly in the functor,
with \emph{relation lifting} (see Definition~\ref{d:rellift}) being the novel,
unifying concept.
\end{enumerate}

\noindent
Let us now look at the construction in more detail.
Before arriving at the actual definition of the regular, nondeterministic
automaton $\sh{\bbA}$, we discuss and motivate the ideas mentioned above.

\subsection{Relations and traces}

We start with motivating the use of \emph{binary relations} on $A$ as the
states of $\sh{\bbA}$.

First, the construction is based on the principle that $\eloi$ should be
prepared to counter many of $\abel$'s moves \emph{simultaneously}.
Intuitively, then, it would be a natural move to construct an automaton
$\bbA^{*}$ taking subsets of $A$ as its states.
Such a \emph{macro-state} would represent the set of states that $\bbA$ could 
be in and that $\eloi$ should be able to somehow handle simultaneously.
Building on this intuition we could proceed to give a precise definition
of the automaton $\bbA^{*}$, generalizing the subset construction for 
automata over finite words.

Continuing along these lines, we might establish a tight link between the 
basic positions 
\begin{equation}
\label{eq:mt1}
(s_{0},\{\ai\})(s_{1},B_{1})\ldots (s_{k},B_{k})
\end{equation}
of a partial match of $\G(\bbA^{*},\bbS)$, and a \emph{collection} of partial 
matches 
\begin{equation}
\label{eq:mt2}
(s_{0},\ai)(s_{1},a_{1})\ldots (s_{k},a_{k})
\end{equation}
in $\G(\bbA,\bbS)$ such that every $a_{i}$ is an element of $B_{i}$.
This link would then naturally extend to infinite matches.

Unfortunately however, we encounter a difficulty when we try to formulate
an adequate acceptance condition for $\bbA^{*}$.
The problem is that, just on the basis of an infinite sequence of subsets of
$\bbA$, we may fail to make some subtle but crucial distinctions.
The point is that the acceptance condition for $\bbA^{*}$ should declare
$\eloi$ as the winner of the match (\ref{eq:mt1}) if and only if she is the
winner of each associated match of the form (\ref{eq:mt2}).
But we may mistakenly declare $\abel$ as the winner of (\ref{eq:mt1}) on the
basis of a sequence of the form (\ref{eq:mt2}), which satisfies $a_{i}\in
B_{i}$ for each $i$, and meets the winning conditions for $\abel$, but
\emph{which did not come about as an actual match of $\G(\bbA,\bbS)$ 
associated with (\ref{eq:mt1})}.
(As many readers will have recognized, this is exactly the problem one
faces when transforming a nondeterministic word automaton into an
equivalent deterministic one, and explains why the Safra construction is
so much more involved than the power set construction.)

An elegant way to avoid this problem is to use \emph{binary relations over $A$}
rather than subsets.
When considered statically, the relation $R$ simply represents the macrostate
$\Ran(R)$ (that is, the \emph{range} of $R$).
The additional structure of binary relations comes into play when we look at
infinite sequences:
The key notion of a \emph{trace} through a sequence of binary relations
allows us to make the required subtle distinctions referred to above.

\begin{definition}
Given an infinite word $\rho = R_{1}R_{2}R_{3}\ldots$ over the set $\relA$
of binary relations over a set $A$, a \emdef{trace} through $\rho$ is a finite
or infinite $A$-word $\al = a_{0}a_{1}a_{2}
\ldots a_{k}$ or  $\al = a_{0}a_{1}a_{2}\ldots\ $
such that
$a_{i}R_{i+1}a_{i+1}$ for all $i$ (respectively, for all $i < k$).

Relative to a priority map $\Par$ on $A$, call a trace $\al$ \emdef{bad} if it
is infinite and the maximum priority occurring infinitely often on $\al$, is
an odd number.
Let $\NOT_{\Par}$ denote the set of infinite $\relA$-words that contain no
bad traces relative to $\Par$.
\end{definition}

Since $\relA$ is the carrier set of $\sh{\bbA}$, infinite matches of the
acceptance game $\G(\sh{\bbA},\bbS)$ induce $\relA$-streams, whereas traces
on such a sequence may be linked to associated matches of the acceptance game
for $\bbA$.
Thus it will be natural to declare $\eloi$ as the winner of a
$\G(\sh{\bbA},\bbS)$-match if there is no bad trace on the induced infinite
$\relA$-sequence, since bad traces correspond to $\G(\bbA,\bbS)$ matches 
that are won by $\abel$.
This explains the acceptance condition $\NOT_{\Om}$ of the automaton
$\sh{\bbA}$.

\begin{remark}
Note, however, that this acceptance condition is not a parity condition.
This means that the automaton $\sh{\bbA}$ is not the automaton $\bbA^{\bl}$
referred to in the statement of Theorem~\ref{t:main}.
In order to transform $\sh{\bbA}$ into a nondeterministic parity automaton,
we need to prove that $\sh{\bbA}$ is a regular automaton, that is, its
acceptance condition $\NOT_{\Om}$ is an $\omega$-regular language, recognized
by a word automaton.
This result will be proved in section~\ref{s:regaut}.
\end{remark}

\subsection{Normalized strategies}

In the sequel it will be convenient to work with so-called \emph{normalized}
strategies for $\eloi$.
Intuitively, these are positional strategies for $\eloi$ that provide her, in
the ``dynamic'', second half of each round of the game $\G(\bbA,\bbS)$, with
a relation between $S$ and $A$ that does not depend on $\abel$'s move in the
first part of the round.

For more details, suppose that $\Phi$, together with $Y: S \times \F A \to
\Sb{S \times A}$, is a positional strategy for $\eloi$ in $\G$, and consider
a basic position $(s,a)$.
Here first $\eloi$ chooses an element $\Phi_{s,a}\in\De(a)$, and then, for
every choice $\phi \in \Phi_{s,a}$, she can choose a separate relation
$Y_{s,\phi} \sse \Sb{S \times A}$.
If she uses a normalized strategy however, then her choice of $Y_{s,\phi}$
must be independent of $\phi$; it may however depend on the earlier basic
position $(s,a)$.
Formally then, we model the dynamic part of a normalized strategy as a map
$Z: S \times A \to \power(S \times A)$.
Intuitively, $Z_{s,a}$ consists of those elements $(t,b)$ that $\eloi$ may 
expect as the next basic position after $(s,a)$.

For technical reasons it will be convenient to add one more condition to
the definition of a normalized strategy:
In the static part of the game, we require $\eloi$ to head for an
immediate win if there is one.
More precisely, consider a state $a \in A$ such that $\De(a)$ contains the
empty set $\nada$ as a choice. 
Clearly, for such an $a$ at any position $(s,a)$ $\eloi$ may choose $(s,\nada)$
as the next position and win immediately, since $\abel$ cannot choose an
element from the empty set. 
A normalized strategy requires $\eloi$ to indeed choose $\nada$ in such a 
position.

Let us now first give the formal definition of a normalized strategy.

\begin{definition}
Given an alternating $\F$-automaton $\bbA$ and an $\F$-coalgebra $\bbS$, a
\emdef{normalized} strategy for $\eloi$ in the game $\G(\bbA,\bbS)$ is a pair
$(\Psi,Z)$ where $\Psi: S \times A \to \pwA \F A$ and $Z: S \times A \to
\power(S \times A)$ are such that $\Psi_{s,a} = \nada$ if $\nada \in \De(a)$.
\end{definition}

If $\eloi$ uses a normalized strategy $(\Phi,Z)$, we can present the
interaction pattern of the players per round as follows, starting at a basic
position $(s,a) \in S \times A$:
\begin{enumerate}[$\bullet$]
\item 
$\eloi$ picks and plays $\Phi_{s,a} \in \De(a)$
\item 
$\eloi$ chooses a relation $Z_{s,a} \sse S \times A$;
\item
$\abel$ picks and plays a $\phi \in \Phi_{s,a}$;
\item
$\eloi$ plays $Z_{s,a}$;
\item
$\abel$ picks and plays a pair $(t,b) \in Z_{s,a}$ as the next basic position.
\end{enumerate}
The resulting interaction pattern is indeed of the earlier announced 
\emph{strategic normal form} `$\eloi\eloi\abel\abel$' rather than of
the form `$\eloi\abel\eloi\abel$'.

The following proposition states that without loss of generality we may 
always assume that $\eloi$'s winning strategies are normalized.

\begin{proposition}
\label{p:ns}
Fix an alternating parity $\F$-automaton $\bbA$ and an $\F$-coalgebra $\bbS$.
Then there is a normalized strategy for $\eloi$ which is winning from every
position $(s,a) \in \Win_{\eloi}(\G(\bbA,\bbS))$.
\end{proposition}

\begin{proof}
Let $\bbA$ and $\bbS$ be as in the proposition.
By the historyfree determinacy of the parity game $\G = \G(\bbA,\bbS)$ we may
assume the existence of a positional strategy 
\[
(\Phi: S \times A \to \pfa, Y: S \times \F A \to \power(S \times A))
\]
which is winning for $\eloi$ from every position $(s,a) \in
\Win_{\eloi}(\G(\bbA,\bbS))$.

Define the strategy $(\Psi,Z)$ as follows:
\begin{eqnarray*}
\Psi_{s,a} &:=& 
   \left\{ \begin{array}{ll}
   \nada & \mbox{if } \nada \in \De(a)
   \\ \Phi_{s,a} & \mbox{otherwise}
   \end{array}\right.
\\ Z_{s,a} &:=& 
   \bigcup_{\phi\in\Psi_{s,a}} Y_{s,\phi}.
\end{eqnarray*}

Consider \emph{one round} of the game $\G$ starting at a winning position
for $\eloi$. 
We claim that either $\eloi$ wins already during this round, or else the
match arrives at a new basic position that could also have been reached if
$\eloi$ had played her original strategy $(\Phi,Y)$.
From this claim one may derive that the strategy $(\Psi,Z)$ guarantees
$\eloi$ to win any match of $\G$ starting at a position in 
$\Win_{\eloi}(\G(\bbA,\bbS))$.

To prove our claim, take a position $(s,a) \in \Win_{\eloi}(\G(\bbA,\bbS))$.
To start with, it is easy to check that $\Psi_{s,a}$ is a legitimate move for
$\eloi$.
If $\Psi_{s,a} = \nada$, then $\eloi$ wins immediately.
So suppose otherwise, and let $\abel$ pick an element $\phi \in \Psi_{s,a}
= \Phi_{s,a}$.
It follows from $(s,a) \in \Win_{\eloi}(\G(\bbA,\bbS))$ 
and the fact that $(\Phi,Y)$ is a winning strategy for $\eloi$, that 
$(\si(s),\phi) \in \Fb Y_{s,\phi}$.
Hence, by the monotonicity of $\Fb$ (see Fact~\ref{fact:wpb}) and the
definition of $Z_{s,a}$, we find that $(\si(s),\phi)
\in \Fb Z_{s,a}$, so that $Z_{s,a}$ is a legitimate answer to $\abel$'s move
$\phi$.
If $Z_{s,a} = \nada$ then $\eloi$ wins immediately, otherwise $\abel$ may
finish the round by picking an element $(t,b) \in Z_{s,a}$.

It remains to be shown that such an element $(t,b) \in Z_{s,a}$ could also
have been obtained if $\eloi$ had played her original strategy $(\Phi,Y)$.
But it follows by definition of $Z_{s,a}$ that $(t,b) \in Y_{s,\psi}$ for
some $\psi \in \Psi_{s,a}$.
Thus $\Psi_{s,a} \neq \nada$, and so $\Psi_{s,a} = \Phi_{s,a}$.
Hence the position $(t,b)$ could have been reached in the scenario where
$\eloi$ had played $\Phi_{s,a}$, followed by $\abel$ picking $\psi \in
\Phi_{s,a}$, $\eloi$ choosing the move $Y_{s,\psi}$, and, finally, $\abel$ 
playing the pair $(t,b) \in Y_{s,\psi}$.
The only thing left to verify here is the legitimacy of the move
$Y_{s,\psi}$, but this is immediate by the assumption that $Y$ is part of 
a winning strategy for $\eloi$.
\end{proof}

\subsection{Normalized strategies and binary relations}

In order to see how the ideas of the previous two subsections fit together,
consider again a round of the acceptance game in which $\eloi$ uses a 
normalized strategy $(\Phi,Z)$:
\begin{enumerate}[$\bullet$]
\item 
$\eloi$ plays $\Phi_{s,a} \in \De(a)$ (and chooses $Z_{s,a} \sse S \times A$)
\item
$\abel$ plays $\phi \in \Phi_{s,a}$; 
\item
$\eloi$ plays $Z_{s,a} \sse S \times A$;
\item
$\abel$ plays $(t,b) \in Z_{s,a}$ as the next basic position.
\end{enumerate}
The point about normalized strategies is that in fact the first two moves
of such a round are only of interest if they lead to an immediate end of the
match in that one of the players gets stuck, i.e., if either $\De(a)$ or 
$\Phi_{s,a}$ is empty.
For infinite matches, the only relevant interaction is between $\eloi$
choosing binary relations between $S$ and $A$, and $\abel$ choosing elements
of those relations.
Let us look at this interaction in a bit more detail.

Recall that $Z_{s,a}$ contains those elements $(t,b)$ that $\eloi$ `expects'
as the next basic position after $(s,a)$.
Thus the dynamic part $Z$ of $\eloi$'s strategy induces a tree of basic
positions for $\abel$ to choose from.
We will now reorganize this tree, as follows.

First observe that, given a relation $Z_{s,a}$, for a single $t \in S$, there
may be \emph{many} elements $b \in A$ such that $(t,b) \in Z_{s,a}$.
These are the states that $\eloi$ should prepare for to meet `simultaneously'
at the point $t \in S$, and that may be grouped together in a `macro-state',
as discussed earlier on.
But then inductively, at $s$, the state $a$ might already have been one of 
many parallel states in some macro-state.
Here it starts making a lot of sense to involve binary relations: Instead of
having macro-states $\{ b \in A \mid (t,b) \in Z_{s,a} \}$, for each $a \in A$,
we consider the binary relation 
\begin{equation}
\label{eq:zeta1}
\zeta_{s}(t) := \{ (a,b) \in A \times A \mid (t,b) \in Z_{s,a} \}.
\end{equation}
Formally, we may represent the dynamic part $Z$ of a normalized strategy
as a map $\zeta: S \to (S \to \Pow(A \times A))$, 
where $\zeta_{s}$ is a map assigning
a binary relation on $A$ to each $t \in S$.
The connection between $Z$ and $\zeta$ is given by
\begin{equation}
\label{eq:aaa}
(t,b) \in Z_{s,a} \iff (a,b) \in \zeta_{s}(t).
\end{equation}
It is not hard to show that (\ref{eq:aaa}) induces a natural bijection
\begin{equation}
\label{eq:aab}
S \times A \to \power (S \times A)
\;\cong\; S \to (S \to \power (A \times A)).
\end{equation}
In fact, using currying ($P \to (Q \to R) \cong (P \times Q) \to R)$ and
exponentiation ($\power(Q) \cong Q \to 2$), it is very easy to see why
(\ref{eq:aab}) must hold:
\begin{equation}
S \times A \to \power (S \times A)
\;\cong\;
(S \times S \times A \times A) \to 2
\;\cong\;
S \to (S \to \power (A \times A)).
\end{equation}
Conversely, an explicit way of obtaining $Z$ from $\zeta$ is as follows.
Let the map $\eva: \Pow(A \times A) \to \power A$ be given by $\eva: R \mapsto 
R[a]$, and recall that the graph $\{ (x,fx) \mid x \in X \} \sse X \times Y$
of a function $f: X \to Y$ is denoted as $\Graph{f}$.
It is then easy to see that 
\begin{equation}
\label{eq:zeta2}
Z_{s,a} = \Graph{\zeta_{s}} \circ \Graph{\eva} \circ \niA.
\end{equation}
Summarizing the above discussion we give the following proposition, of
which we will make heavily use in the sequel.

\begin{proposition}
\label{p:pietje}
For any pair of sets $S$ and $A$, there is a natural bijection between maps
$Z: S \times A \to \power (S \times A)$ and functions
$\zeta: S \to (S \to \power (A \times A))$.
This correspondence is explicitly given by (\ref{eq:zeta1}) and
(\ref{eq:zeta2}) above.
\end{proposition}

\subsection{The definition of $\sh{\bbA}$}

We have already announced that $\sh{\bbA}$ will be a nondeterministic 
regular $\F$-automaton based on the collection $\relA$ of binary relations
on $A$, and with acceptance condition of the form $\NOT_{\Om}$, where $\Om$
is the parity condition of $\bbA$.
Thus to complete the definition of the automaton $\sh{\bbA}$ it suffices to
give the transition structure $\shDe: \relA \to \power(\F\relA)$.

Roughly speaking, it works like this.
Earlier on we already briefly mentioned that, intuitively, a relation $R \in 
\relA$ represents the macrostate $\Ran(R) \sse A$.
Now suppose we consider the static part of a strategy for $\eloi$ at a
certain point of the coalgebra (see the discussion following
Definition~\ref{def:Fautomaton} for a division of a round of the acceptance
game into a static and a dynamic part).
In order to handle \emph{each} of the challenges $a \in \Ran(R)$, $\eloi$
needs to come up with a \emph{family}
\begin{equation}
\label{eq:fam}
\{ \Phi(a) \in \power(\F A) \;\mid\; a \in \Ran(R) \}
\mbox{ such that } \Phi(a) \in \De(a) \mbox{ for each } a \in \Ran(R).
\end{equation}
Our definition of $\shDe$ will be such that given $R \in \relA$, the members
of $\shDe(R)$ are those elements $\Pi \in \F\relA$ that are in a natural
correspondence with such a family.

The key to understanding this `natural correspondence' is the notion of 
\emph{$\F$-redistribution}, which links sets of the form $\F\power A$ and 
$\power\F A$.
For an introduction to this notion, first consider the membership relation
$\inA$ on $A$.
Since $\inA$ is a binary relation between $A$ and $\power(A)$, we may lift
it to a relation $\Fb\inA$ between $\F A$ and $\F \power A$.
The relation $\Fb\inA$ is like the membership relation `behind an $\F$-veil'.
Now suppose that $\Phi \in \power\F A$ and $\Xi \in \F\power A$ satisfy the
condition that \emph{each} element $\phi$ of $\Phi$ is such an `$\F$-member'
of $\Xi$, i.e., $(\phi,\Xi) \in \Fb\inA$.
In such a case we call $\Xi$ an $\F$-redistribution of $\Phi$, and it makes
sense to think of $\Xi$ as a representation of $\Phi$ as a set of type
$\F\power A$.

\begin{definition}
Given a set $A$ with membership relation $\inA \sse A \times \power (A)$, we 
call $\Xi \in \F\power A$ an \emdef{$\F$-redistribution of $\Phi \in 
\power\F A$}, or say that $\Xi$ \emph{redistributes} $\Phi$, if $(\phi,\Xi) 
\in \Fb\inA$ for all $\phi \in \Phi$.
\end{definition}

\begin{example}
For the binary tree functor $\BT_{C}$, an element $(c,a^{l},a^{r}) \in
\BT_{C}A$ is an $\BT_{C}$-member of an object $(d,A^{l},A^{r}) \in \BT_{C}(A)$
if $c=d$, $a^{l}\in A^{l}$ and $a^{r} \in A^{r}$.
Hence $(d,A^{l},A^{r}) \in \BT_{C}(A)$ is a $\BT_{C}$-redistribution of 
the set $\{ (c_{i},a^{l}_{i},a^{r}_{i}) \mid i \in I \}$ iff $c_{i}=d$,
$a^{l}_{i}\in A^{l}$ and $a^{r}_{i} \in A^{r}$, for each $i \in I$.

For the power set functor $\K$, an object $X \in \K A = \Sb{A}$ is a
$\K$-element of an object $\mathcal{B} \in \K\Sb{A} = \Sb{\Sb{A}}$ iff
$X \sse \bigcup \mathcal{B}$ and $X \cap B \neq \nada$ for all $B \in
\mathcal{B}$.
So $\mathcal{B} \in \K\Sb{A}$ is a $\K$-redistribution of $\mathcal{X}
\in \Sb{\K A}$ iff $\bigcup \mathcal{X} \sse \bigcup \mathcal{B}$ and 
$X \cap B \neq \nada$ for all $X \in \mathcal{X}$ and all $B \in
\mathcal{B}$.
\end{example}

\begin{remark}
In \cite{jacobs:distlaws} Jacobs shows that for every weak pullback
preserving functor $\F$ there is a so-called {\em distributive law} 
$\lambda: \F \K \Rightarrow \K \F$ of $\F$ over the power set monad,
i.e.\ $\lambda$ is a natural transformation that preserves the monad
structure.
This distributive law is defined using the relation lifting of the 
$\in$-relation. 
Therefore there is a close connection between Jacobs's law and our
$F$-redistributions: $\Xi \in \F \power A$ is an $\F$-redistribution of
$\Phi \in \power \F A$ iff $\Phi \subseteq \lambda_A(\Xi)$.  
\end{remark}

We are now almost ready for the definition of $\shDe$.
For the final step, recall that for any element $a \in A$ we may go from 
$\relA$ to $\power(A)$ using the \emph{evaluation map} 
\[
\eva: R \mapsto R[a].
\]
Thus $\F \eva: \F\relA \to \F\power(A)$.
The function $\F\eva$ enables us to link potential elements $\Pi \in \shDe(R)$
to redistributions in $\F\power (A)$ of objects $\Phi(a) \in \De(a)$.

\begin{definition}
Let $\F$ be a set functor that preserves weak pullbacks, and let
$\bbA = \struc{A,\ai,\De,\Par}$ be an alternating $\F$-automaton.
Then the automaton $\sh{\bbA}$ is defined as the structure
\[
\sh{\bbA} := \struc{\relA, \Ri, \shDe, \NOT_{\Par}}, 
\]
where $\relA := \power(A \times A)$ is the collection of binary relations
over $A$, $\Ri = \{ (\ai,\ai) \}$, $\shDe: \relA \to \pwE\F\relA $ is given
by
\[
\begin{array}{lcll}
\shDe(R) &:=&
   \{ \Pi \in \F\relA \mid &
\forall a \in \Ran(R)\, \exists \Phi(a) \in \De(a)\, 
\\ &&& \ \ \  (\F\eva)(\Pi) \mbox{ is an $\F$-redistribution of }
\Phi(a)
\},
\end{array}
\]
and $\NOT_{\Par} \sse (\relA)^{\om}$ is the set of those infinite sequences
of binary relations that do not contain any bad trace.
\end{definition}

It is obvious that $\sh{\bbA}$ is a nondeterministic $\F$-automaton.
It remains to be shown that $\sh{\bbA}$ is equivalent to $\bbA$.

\subsection{Proof of equivalence}

\begin{proposition}
Let $\F$ be a set functor that preserves weak pullbacks, and let
$\bbA = \struc{A,\ai,\De,\Par}$ be an alternating parity $\F$-automaton.
Then $\sh{\bbA}$ is equivalent to $\bbA$.
\end{proposition}

\begin{proof}
Fix a coalgebra $\bbS$ and a point $s_{0}$ in $\bbS$.
We will prove the following equivalence:
\begin{equation}
\label{eq:zzz}
\bbA \mbox{ accepts } (\bbS,s_{0}) \iff
\sh{\bbA} \mbox{ accepts } (\bbS,s_{0}).
\end{equation}
Obviously, both directions of this equivalence will be proved via a comparison
of the two acceptance games $\G := \G(\bbA,\bbS)$ and $\sh{\G} := 
\G(\sh{\bbA},\bbS)$.
\medskip

\noindent
\fbox{$\Rightarrow$}
For the direction from left to right, assume that $\bbA \mbox{ accepts }
(\bbS,s_{0})$.
Then by Proposition~\ref{p:ns}
we may assume that $\eloi$ has a normalized strategy $(\Phi,Z)$ which is
winning for her in the game $\G$ initialized at $(s_{0},\ai)$.
In the sequel we will also make use of the map $\zeta: S \times S \to 
\relA$ that is associated with $Z$ as in Proposition~\ref{p:pietje}.

In order to prove that $(s_{0},\Ri)$ is a winning position for $\eloi$ in
the game $\sh{\G}$, we let $\eloi$ play according to the following
positional strategy $(\Pi,Q)$.
The \emph{static} part $\Pi: S \times \relA \to \F\relA$ of this strategy
is given by
\[
\Pi_{s,R} := (\F\zeta_{s})(\si(s)),
\]
while the \emph{dynamic} part $Q: S \times \F\relA \to \power(S \times \relA)$
is defined as
\[
Q_{s,\Sigma} := \Graph{\zeta_{s}}.
\]
Since the function $Q$ only depends on its first argument, in the sequel we
will simply write $Q_{s}$ instead of $Q_{s,\Sigma}$.

The claims~2 and~3 below state that, playing this strategy, $\eloi$ wins all
finite, respectively, infinite matches.
In the proof of these results we need the following additional claim which,
roughly spoken, states that the strategy defined above is legitimate at any
safe position $(s,R)$ of $\sh{\G}$, and guarantees that the next basic 
position is safe as well.
Here we call a basic position $(s,R)$ of $\sh{\G}$ \emph{safe} if $(s,a) \in 
\Win_{\eloi}(\G)$ for all $a \in \Ran(R)$.

\begin{claimeerst}
Let $(s,R)$ be a position of $\sh{\G}$ such that $(s,a) \in \Win_{\eloi}(\G)$
for all $a \in \Ran(R)$.
Then 

(1) both $\Pi$ and $Q$ provide legitimate moves at $(s,R)$,

(2) $(t,b) \in Z_{s,a}$ for all $(t,R') \in Q_{s}$, all $a \in A$ and all
    $b \in R'[a]$.
\end{claimeerst}

\begin{pfclaim}
The main part of the proof consists in showing that $\Pi:= \Pi_{s,R}$ is a
legitimate move for $\eloi$ at position $(s,R)$.

In order to show that, indeed, $\Pi \in \shDe(R)$, consider an arbitrary 
element $a \in \Ran(R)$.
By assumption, $(s,a) \in \Win_{\eloi}(\G)$.
Recall that $\Phi_{s,a}$ and $Z_{s,a}$ are the moves of $\eloi$ in $\G$
prescribed by $\eloi$'s winning normalized strategy.
Take an arbitrary element $\phi \in \Phi_{s,a}$.
It suffices to prove that 
\begin{equation}
\label{eq:ggg}
(\phi,(\F\eva)(\Pi)) \in (\Fb\inA),
\end{equation}
since this implies that $(\F\eva)(\Pi)$ is an $\F$-redistribution of
$\Phi_{s,a} \in \De(a)$, and thus that $\Pi \in \shDe(R)$, since $a$ was
arbitrary.

It follows from the fact that $Z_{s,a}$ is part of a winning, and thus
legitimate strategy, that 
\begin{equation}
\label{eq:ggh}
(\si(s),\phi) \in \Fb Z_{s,a}.
\end{equation}

Now from $Z_{s,a} = \Graph{\zeta_{s}} \circ \Graph{\eva} \circ \niA$, (see
(\ref{eq:zeta2}) and some elementary properties of relation lifting 
(cf.~Fact~\ref{fact:wpb})), it follows that 
\[
\Fb(Z_{s,a}) = \Graph{\F\zeta_{s}} \circ \Graph{\F\eva} \circ \Fb\niA.
\]
Thus from (\ref{eq:ggh}) and the fact that $\Pi = (\F\zeta_{s}(\si(s))$ is
defined as the \emph{unique} object such that $(\si(s),\Pi) \in 
\Graph{\F\zeta_{s}}$, it is immediate that 
\[
(\Pi,\phi) \in \Graph{\F\eva} \circ \Fb\niA,
\]
which is easily seen to be equivalent to (\ref{eq:ggg}).

To finish the proof of part (1) of the claim, it then suffices to show that 
$Q_{s}$ is a legitimate move at position $(s,\Pi)$ (where still we
write $\Pi = \Pi_{s,R}$).
But this is immediate by the definitions.
The point is that from $\Pi = (\F\zeta_{s})(\si(s))$ we may infer 
$(\si(s),\Pi) \in \Graph{\F\zeta_{s}} = \Fb\Graph{\zeta_{s}} = \Fb
Q_{s}$.

Part (2) of the claim is also straightforward.
Let $a,b \in A$, $t \in S$ and $R' \in \relA$ be such that $(t,R') \in Q$
and $(a,b) \in R'$.
Recall that by definition of $Q$, $(t,R') \in Q$ implies that $R' =
\zeta_{s}(t)$, so by (\ref{eq:aaa}) we have $(a,b) \in R'$ iff $(t,b) 
\in Z_{s,a}$.
\end{pfclaim}

\begin{claim}
As long as $\eloi$ plays her strategy $(\Pi,Q)$, she wins all finite matches
starting at position $(s_{0},\Ri)$.
\end{claim}

\begin{pfclaim}
A straightforward inductive proof using part (2) of Claim~1 shows that any
partial $\sh{\G}$-match $(s_{0},\Ri)(s_{1},R_{1})\ldots(s_{n},R_{n})$ of
$\sh{\G}$ in which $\eloi$ plays her strategy $(\Pi,Q)$ has the property that
\[
(s_{n},b) \in \Win_{\eloi}(\G) \mbox{ for all } b \in \Ran(R_{n}).
\]
Then by part (1) of Claim~1 it follows that $\Pi$ and $Q$ provide legitimate
moves for $\eloi$. 
In other words, she will not get stuck after position $(s_{n},R_{n})$.
\end{pfclaim}

\begin{claim}
As long as she plays her strategy $(\Pi,Q)$, $\eloi$ wins all infinite matches
starting at position $(s_{0},\Ri)$.
\end{claim}

\begin{pfclaim}
Consider an infinite match 
\[
(s_{0},R_{0})(s_{1},R_{1})\ldots
\]
of $\sh{\G}$ in which $\eloi$ plays her strategy $(\Pi,Q)$ (and with
$\Ri = R_{0}$).
In order to show that this match is won by $\eloi$, consider an arbitrary
trace on the sequence $\Ri R_{1}R_{2}\ldots$\ 
It suffices to show that this trace is even.

Clearly the trace is of the form $a_{0}a_{0}a_{1}a_{2}\ldots$ with 
$\ai = a_{0}$, $a_{0}R_{0}a_{0}$ and $a_{i}R_{i+1}a_{i+1}$ for every $i$.
A direct inductive proof, using part (2) of Claim~1, shows that
$(s_{i+1},a_{i+1}) \in Z_{s_{i},a_{i}}$ for every $i$.
From this it is easy to find a match 
\[
(s_{0},\ai)(s_{1},a_{1})\ldots
\]
of $\G$ in which $\eloi$ plays her strategy $(\Phi,Z)$, cf.~the proof of 
Proposition~\ref{p:ns}.
But by assumption, this strategy is \emph{winning} for $\eloi$, so the trace
$a_{0}a_{0}a_{1}a_{2}\ldots$ is indeed \emph{even}.
\end{pfclaim}

Finally, the direction $\Rightarrow$ of (\ref{eq:zzz}) is a direct consequence
of the Claims~2 and~3.
\medskip

\noindent
\fbox{$\Leftarrow$}
For the direction from right to left, assume that $\sh{\bbA}$ accepts
$(\bbS,s_{0})$.
In other words, we may assume that there is a strategy $f$ which is winning
for $\eloi$ in the game $\sh{\G}$ starting at $(s_{0},\Ri)$.
In order to show that $\bbA$ accepts $(\bbS,s_{0})$, we need to prove that
$(s_{0},\ai)$ is a winning position for $\eloi$ in $\G$.

We will equip $\eloi$ with a strategy $f'$, in the game $\G$ initialized at
$(s_{0},\ai)$, which has the following property.
For any (possibly finite) $f'$-conform match $(s_{0},a_{0})(s_{1},a_{1})\ldots$
of $\G$ with $a_{0}=\ai$, there is an $f$-conform match 
$(s_{0},R_{0})(s_{1},R_{1})\ldots$ of $\sh{\G}$, with $R_{0}=\Ri$, satisfying
the condition that 
\begin{equation}
\label{eq:abcd}
a_{i+1} \in R_{i+1}[a_{i}] \mbox{ for every stage $i$}.
\end{equation}

Hence, the sequence of $\bbA$-states $a_{0}a_{1}a_{2}\ldots$ of such a match 
is a \emph{trace} of the $\relA$-sequence $R_{0}R_{1}R_{2}\ldots$ which we
may associate with an $f$-conform match.
Since $f$ is by assumption winning for $\eloi$, by definition of the winning
condition $\NOT_{\Par}$ of $\sh{\bbA}$, the (maximum parity occurring 
infinitely often on) the trace must be \emph{even}.

This guarantees that she wins all \emph{infinite} matches of the game.
Hence, it suffices to prove that at any finite stage of an $f'$-conform
match, she either immediately, or else she can keep the above condition
for one more round.

Suppose then that $\eloi$ has been able to keep this condition for $k$ steps.
That is, with the partial $\G$-match $(s_{0},a_{0})\ldots (s_{k},a_{k})$ 
(where $\ai=a_{0}$) we
may associate a partial, $f$-conform $\sh{\G}$-match $(s_{0},R_{0})\ldots
(s_{k},R_{k})$ such that $R_{0}=\Ri$ and
\begin{equation}
\label{eq:abc}
a_{i+1} \in R_{i+1}[a_{i}]
\mbox{ for all $i<k$}.
\end{equation}
For notational convenience, write $a = a_{k}$, $R = R_{k}$ and $s = s_{k}$,
so we have $a \in \Ran(R)$.
Let $\Pi \in \F\relA$ and $Q \sse S \times \relA$, respectively, be the
moves dictated by $\eloi$'s winning strategy $f$ in $\sh{\G}$.
It follows from the fact that $f$ is a winning strategy, that $\Pi$ and $Q$
are legitimate moves, that is, $\Pi \in \shDe(R)$ and $(\si(s),\Pi) \in 
\Fb(Q)$.
Then by definition of $\shDe$, and the fact that $a \in \Ran(R)$, there is
some $\Phi \in \De(a)$ such that $(\F\eva)(\Pi)$ is an $\F$-redistribution
of $\Phi$.
This $\Phi$ is the next move of $\eloi$ in the game $\G$.

If $\Phi = \nada$ then $\eloi$ wins right away, in which case we are done
immediately.
So assume that $\Phi \neq \nada$, and suppose that $\abel$ responds to 
$\eloi$'s  move with an object $\phi \in \Phi$.
Then $\eloi$ has to come up with a relation $Y \sse S \times A$ such that
$(\si(s),\phi) \in \Fb(Y)$.
Our suggestion to $\eloi$ is to pick the relation given by 
\[
Y := Q \circ \Graph{\eva} \circ \niA,
\]
or, spelled out,
\[
Y = \{ (t,b) \in S \times A \mid b \in R'[a] \mbox{ for some } R' \in \relA
       \mbox{ with } (t,R') \in Q \}.
\]

If this is a legitimate move for $\eloi$, then we are done.
For, distinguish the following cases.
If $Y = \nada$ then $\abel$ gets stuck so $\eloi$ wins immediately.
But if $Y \neq \nada$, then with any $(s_{k+1},a_{k+1}) \in Y$ that $\abel$
chooses as his next move, by definition we may associate a relation $R_{k+1}
\in \relA$ such that $(a_{k},a_{k+1}) \in R_{k+1}$ and $(s_{k+1},R_{k+1})
\in Q$.
In other words, we have showed that $\eloi$ can indeed maintain the above
mentioned condition (\ref{eq:abc}) for one more round of the game.

Thus it is left to show that $Y$ is a legal move for $\eloi$ in $\G$; that 
is, we must show that 
\begin{equation}
\label{eq:1}
(\si(s),\phi) \in \Fb(Y).
\end{equation}
For this purpose, first observe that the definition of $Y$ and the properties 
of $\Fb$ (cf.~Fact~\ref{fact:wpb}) imply that 
\begin{equation}
\label{eq:2}
\Fb(Y) = \Fb(Q) \circ \Graph{\F\eva} \circ \Fb(\niA).
\end{equation}
Now it follows from the legitimacy of $\Pi$ in the game $\sh{\G}$, that
$(\F\eva)(\Pi)$ is an $\F$-redistribution of $\Phi$, i.e.,
\begin{equation}
\label{eq:hhh}
(\Pi,\phi) \in \Graph{\F\eva} \circ \Fb(\niA).
\end{equation}
From the legitimacy of $Q$ it follows that 
\begin{equation}
\label{eq:hhi}
(\si(s),\Pi) \in \Fb Q.
\end{equation}
But then (\ref{eq:1}) is immediate from (\ref{eq:2}), (\ref{eq:hhh}) and
(\ref{eq:hhi}).
\end{proof}

\section{Regular automata}
\label{s:regaut}

\noindent In this section we look in detail at some of the acceptance
conditions of coalgebra automata.  Recall that in the case of an
acceptance game $\G(\bbA,\bbS)$, the winner of any infinite match is
determined by the infinite sequence of $\bbA$-states
\[
\ai a_{1} a_{2} \ldots
\]
that is induced by the match.
More specifically, the acceptance condition of the automaton $\bbA$ is of 
the form $L \sse L^{\om}$, i.e., an $\om$-language over the set $A$ of 
$\bbA$-states.
In many cases, the set $L$ has a fairly low complexity.
For instance, in the case of parity automata, the criterion whether an
$A$-stream $\al$ belongs to $L$ or not is given in terms of the set
$\Inf(\al)$ of those states that occur infinitely often in $\al$.

An interesting class of automata is given by those in which the acceptance
condition is a so-called \emph{$\om$-regular language}, that is, a subset
$L \sse A^{\om}$ that is itself recognized by some word automaton.

\begin{remark}
For readers that are not familiar with the theory of automata operating
on infinite words, we summarize the definitions here.
Fix an alphabet $C$.

A \emph{$C$-stream} is an infinite $C$-word $\ga = c_{0}c_{1}c_{2}\ldots$
A \emph{nondeterministic $C$-automaton} is a quadruple $\bbA = 
\struc{A,\ai,\De,\Acc}$, where $A$ is a finite set, $\ai \in
A$ is the \emph{initial state} of $\bbA$, $\De: A \times C \to \power(A)$ its
\emph{transition function} of $\bbA$, and $\Acc \sse A^{\om}$ its
\emph{acceptance condition}.
Such an automaton is \emph{deterministic} if $\De(a,c)$ is a singleton
for each $a \in A$ and $c \in C$.

A \emph{run} of a deterministic automaton $\bbA = \struc{A,\ai,\De,\Acc}$ on
an $C$-stream $\ga = c_{0}c_{1}c_{2}\ldots$ is an infinite $A$-sequence 
\[
\rho = a_{0}a_{1}a_{2}\ldots
\]
such that $a_{0}= \ai$ and $a_{i+1} \in \De(a_{i},c_{i})$ for every $i \in
\omega$.
Note that such a run is unique if $\bbA$ is deterministic.

A nondeterministic $C$-automaton $\bbA = \struc{A,\ai,\De,\Acc}$ \emph{accepts}
an $C$-stream $\ga$ if there is a successful run of $\bbA$ on $\ga$.
The set of those streams is denoted by $L_{\om}(\bbA)$.
A set $L \sse C^{\om}$ is called \emph{$\om$-regular} if  $L = L_{\om}(\bbA)$
for some $C$-automaton $\bbA$ with a parity acceptance condition.

A key result in the theory of stream automata states that the every
nondeterministic parity automaton can be transformed into an equivalent
deterministic parity automata.
That is, every $\om$-regular language $L$ is of the form $L_{\om}(\bbA)$ 
for some \emph{deterministic} parity automaton $\bbA$.
\end{remark}

\begin{definition}
An $\F$-automaton $\bbA = \struc{A,\ai,\De,L}$ is called \emdef{regular} if $L 
\sse A^{\om}$ is an $\om$-regular language.
\end{definition}

It is sometimes attractive to use regular automata because the acceptance
condition may be easier or more intuitive to formulate in the form of
an $\om$-regular language than as a parity condition.
An important example of this was given in the previous section where the
nondeterministic automaton $\sh{\bbA}$ was provided with a regular
acceptance condition.
Nevertheless the recognizing power of (nondeterministic) regular automata is
not strictly greater than that of (nondeterministic) parity automata:
Theorem~\ref{t:regaut}, the main result of this section, states that
every regular nondeterministic automaton can be replaced with an
equivalent nondeterministic parity automaton.
As a consequence of this result, we may use regular automata as a handy,
auxiliary notion in the theory of coalgebra automata.

The key idea underlying the proof of Theorem~\ref{t:regaut} is the
construction of a so-called \emph{wreath product}.
Given a regular nondeterministic $\F$-automaton with state set $A$, and
a parity word automaton $\bbW$ operating on infinite $A$-words, we define
the wreath product as some nondeterministic parity $\F$-automaton.
Informally this automaton $\bbB \odot \bbW$ runs the automaton $\bbB$ on a
given pointed $\F$-coalgebra and feeds the resulting sequence of automata
states into the automaton $\bbW$. 

\begin{definition}
Let $\F$ be a set functor, and 
let $\bbB = \struc{B,b_I,\De,L}$ be a nondeterministic $\F$-automaton, and 
let $\bbW=\struc{W,w_{I},\delta:W \times B \to W,\Omega}$ be a deterministic 
parity word automaton.

Let, for $w \in W$, the map $\de_{w}: B \to B \times W$ be defined by putting
\[
\de_{w}(b) := (b,\de(w,b)).
\]
Using this map, we define the element $\ai \in B \times W$ and the map
$\Gamma: B \times W \to \Pow \F (B \times W)$ be given by
\begin{eqnarray*}
   a_{I}     &:=& (b_I, \de(w_{I},b_{I}))
\\ \Gamma(b,w) &:=& \{ (\F \delta_w)(\phi) \mid \phi \in \De(b) \}.
\end{eqnarray*}
The map $\Psi: B \times W \to \om$ is defined as $\Psi = \Om \circ\pi_{2}$, 
that is, 
\[
\Psi(b,w) := \Om(w).
\]
Finally, the nondeterministic $\F$-automaton 
\[ 
\bbB \odot \bbW = \struc{B \times W, \ai, \Gamma, \Psi}
\]
is the \emdef{wreath product} of $\bbB$ and $\bbW$.
\end{definition}

As we will see now, if the acceptance condition of $\bbA$ is actually the 
language recognized by $\bbW$, then the two automata $\bbA$ and $\bbA\odot
\bbW$ are equivalent.


\begin{theorem}
\label{t:regaut}
Let $\F$ be a set functor that preserves weak pullbacks,  
let $\bbB = \struc{B,b_I,\De,L}$ be a nondeterministic regular
$\F$-automaton and let $\bbW$ be a deterministic parity automaton such
that $L$ is the language accepted by $\bbW$.
Then $\bbB$ and $\bbB\odot\bbW$ are equivalent.
\end{theorem}

\begin{proof}
\newcommand{\Godot}{\G^{\odot}}
Let $\G$ and $\Godot$ be the acceptance games $\G := \G(\bbB,\bbS)$ and 
$\Godot := \G(\bbB\odot\bbW,\bbS)$, respectively, and fix a pointed 
$\F$-coalgebra $(\bbS,s_{0})$.
Our aim is to prove the following equivalence:
\begin{equation}
\label{eq:regaut}
\bbB \mbox{ accepts } (\bbS,s_{0}) \mbox{ iff }
\bbB\odot\bbW \mbox{ accepts } (\bbS,s_{0}).
\end{equation}
In both cases our proof consists in showing that $\eloi$ can mimic the 
match in one game by a correlated match of the other game.
Here we call a partial $\Godot$-match $\pi$ \emph{correlated} to a partial
$\G$-match $\pi'$ if (the respective sequences of basic positions in) $\pi$
and $\pi'$ are of the following form:
\begin{eqnarray*}
   \pi  &=& (s_{0},(b_{0},w_1)) (s_1,(b_1,w_2)) \ldots (s_n,(b_n,w_{n+1})) 
\\ \pi' &=& (s_{0},b_{0}) (s_1,b_1) \ldots (s_n,b_n) 
\end{eqnarray*}
where $w_{i+1} = \de(w_{i},b_{i})$ for each $i$.
A similar definition applies to full matches.

It is not hard to see that infinite correlated matches have the same winner,
i.e., if the infinite matches $\pi$ and $\pi'$ are correlated, then $\eloi$
wins $\pi$ iff she wins $\pi'$.
The key observation here is that if $\pi = (s_{0},(b_{0},w_1)) (s_1,(b_1,w_2))
\ldots$ is correlated to $\pi' =  (s_{0},b_{0}) (s_1,b_1) \ldots$, then 
$w_{0}w_{1}\ldots$ is the run of $\bbW$ on the infinite $B$-word
$b_{0}b_{1}\ldots$
\medskip

We now turn to the proof of (\ref{eq:regaut}).
For the direction from left to right, assume that $\bbB$ accepts 
$(\bbS,s_{0})$, that is, assume that $\eloi$ has a winning strategy in
the game $\G$ initiated at $(s_{0},b_{0})$.
In order to show that $\bbB\odot\bbW$ accepts $(\bbS,s_{0})$, we need to
equip her with a winning strategy in the game $\Godot$ starting from the
position $(s_{0},(b_{0},w_{1}))$.
We will show that with the running $\Godot$-match, $\eloi$ can maintain a
correlated $\G$-match in which she plays her winning strategy.
Then by the fact that infinite correlated matches are won by the same
player, she is guaranteed to win all infinite matches.
Hence it suffices to prove inductively that if she has maintained the
shadow match for $n$ rounds, she either directly wins the $\Godot$-match in
the next round, or else she can maintain the shadow match for one more round.

Assume then that with the partial $\Godot$-match 
\[
\pi = (s_{0},(b_{0},w_1)) (s_1,(b_1,w_2)) \ldots (s_n,(b_n,w_{n+1})) 
\]
she has associated a correlated partial $\G$-match
\[
\pi'= (s_{0},b_{0}) (s_1,b_1) \ldots (s_n,b_n) 
\]
which is conform her winning strategy in $\G$.
Suppose that this winning strategy tells her to choose $\phi \in \De(b_{n})$,
followed by the relation $Y \sse S \times B$ satisfying $(\si(s_{n}),\phi) 
\in \Fb Y$.
Then in the $\Godot$-match $\pi$ she chooses $(F \de_{w_{n+1}})(\phi)$, 
followed by the relation $Z := Y \circ \Graph{F \de_{w_{n+1}}}$.
The legitimacy of these moves is immediate by the definitions.

Clearly if $Z = \nada$, $\eloi$ wins immediately, so assume otherwise, and
suppose that $\abel$ picks a pair $(s,(b,w))$ as the next basic position 
continuing $\pi$.
Then by definition of $Z$ we have $(s,b) \in Y$ and $w = \de(w_{n+1},b)$.
Hence in $\G$ we could have arrived at the position $(s,b)$ if in $\pi'$,
$\eloi$ had chosen $\phi$ and $Y$. 
And since $w = \de(w_{n+1},b)$, the two partial matches $\pi(s,(b,w))$
and $\pi'(s,b)$ are correlated. 
In other words, $\eloi$ has indeed maintained the required condition for
one more round.
\medskip

For the other direction of (\ref{eq:regaut}), assume that $\eloi$ has a
winning strategy in the acceptance game $\Godot$ initiated at $(s_{0},
(b_{0},w_{1}))$.
It suffices to show that in the game $\G$ starting at $(s_{0},b_{0})$,
$\eloi$ has a winning strategy. 
Analogously to the proof for the other direction, we will show that, round
by round, $\eloi$ can maintain a shadow match in $\Godot$ which is 
correlated to the running $\G$-match, and conform her supposed winning
strategy.

More precisely, inductively assume that with the partial $\G$-match
\[
\rho = (s_{0},b_{0}) (s_1,b_1) \ldots (s_n,b_n) 
\]
she has associated a correlated partial $\Godot$-match conform her
winning strategy:
\[
\rho' = (s_{0},(b_{0},w_1)) (s_1,(b_1,w_2)) \ldots (s_n,(b_n,w_{n+1})) 
\]
Consider the moves suggested by the winning strategy in $\Godot$ when the
partial match $\rho'$ arrives at $(s_n,(b_n,w_{n+1}))$.
$\eloi$ first picks an element from $\Gamma(s_n,(b_n,w_{n+1}))$, which by
definition of $\Gamma$ is of the form $(F \de_{w_{n+1}})(\phi)$ with $\phi \in
\De(b)$, followed by a relation $Z \sse S \times (B \times W)$ such that 
$(\si(s_{n}),(F \de_{w_{n+1}})(\phi)) \in \Fb Z$.
Clearly then the pair $(\si(s_{n}),\phi)$ belongs to the set $\Fb Z \circ
\converse{(\Graph{F \de_{w_{n+1}}})}$, which by properties of relation lifting 
(see Fact~\ref{fact:wpb}) is identical to the relation $\Fb (Z \circ 
\converse{(\Graph{\de_{w_{n+1}}})})$.
This means that $\eloi$ may legitimately continue the $\G$-match $\rho$
with the moves $\phi$ and $Y:= Z \circ \converse{(\Graph{\de_{w_{n+1}}})}$.

Suppose that $\abel$ responds to these moves by playing a pair $(s,b) \in Y$.
By definition of $Y$ there must be a pair $(b',w) \in B \times W$ such 
that $(s,(b',w)) \in Z$ and $\de_{w_{n+1}}(b) = (b',w)$.
From this it is immediate that $b' = b$ and $w = \de(w_{n+1},b)$.
But from $(s,(b,w)) \in Z$ it follows that in $\Godot$, $\abel$ may respond
to $\eloi$'s move $Z$ by picking $(s,(b,w))$ as the next basic position.
Then from $w = \de(w_{n+1},b)$ it follows that the partial matches
$\rho(s,b)$ and $\rho'(s,(b,w))$ are correlated, and since the continuation
of $\rho'$ was conform $\eloi$'s winning strategy, we are done.
\end{proof}

\begin{remark}
Our construction does not use the fact that the word automaton $\bbW$ is a
{\em parity} word automaton. We could use any other acceptance condition on
infinite words, such as a B\"uchi, a Muller or a Rabin condition. In these
cases the resulting wreath product automaton would be an $\F$-automaton with
B\"uchi, Muller or Rabin condition respectively.
\end{remark}

\section{Closure properties}
\label{s:3}
\label{s:closprop}

\noindent In this section we prove Theorem~\ref{t:main} and
Theorem~\ref{t:clos}.  First we combine the results of the previous
two sections in order to show that every parity $\F$-automaton can be
transformed into an equivalent nondeterministic one.  After that we
will see that the class of nondeterministically recognizable languages
is closed under taking union and projection, whereas the class of
recognizable languages is shown to be closed under union and
intersection.  Combined with Theorem~\ref{t:main}, this suffices to
prove Theorem~\ref{t:clos}.

\subsection{The main theorem}

In the previous sections we saw how to transform an arbitrary 
parity $\F$-automaton $\bbA$ into a nondeterministic regular 
$\F$-automaton $\sh{\bbA}$. Furthermore
we showed how to use the wreath product construction in order
to transform a given regular $\F$-automaton into an equivalent
$\F$-automaton with parity acceptance condition. We will
combine these facts in order to prove that for 
every parity $\F$-automaton we can effectively construct an equivalent
nondeterministic parity $\F$-automaton. To begin with we have a closer 
look at the size of the word automaton $\bbW$ that witnesses 
the fact that the acceptance condition of $\sh{\bbA}$ is 
regular.

\begin{proposition}\label{prop:wordcomp}
Let $\F$ be a set functor, and let $\bbA=\struc{A,\ro{a},
\De,\Omega}$ be a parity
$\F$-automaton with $n$ states and index $k$. 
Then we can construct a deterministic parity $\relA$-word automaton 
$\bbW=\struc{W,\ro{w},\delta,\Omega'}$, 
such that $\bbW$ accepts $\alpha = R_0 R_1 
R_2 \ldots \in ({\relA})^\omega$ iff $\alpha$ contains no bad trace. 
This automaton has $2^{\mathcal{O}(nk \log (nk))}$ states and index 
$\mathcal{O}(nk)$.
\end{proposition}

\begin{proof}  
The construction of $\bbW$ is done in four steps:
\begin{enumerate}[\hbox to6 pt{\hfill}]
\item\noindent{\hskip-10 pt\bf Step 1:}\ We define a nondeterministic $\relA$-word automaton $\bbB_1
  \coloneqq \struc{A,\ro{a},\delta_1,\Omega^{+1}}$ where we put 
  $\delta_1(a,R) 
  \coloneqq R[a]$ for all $a \in A$ and all $R \in \relA$, and $\Omega^{+1}(a)
  \coloneqq \Omega(a) + 1$. 
  A straightforward argument shows that $\bbB_1$ accepts a word $\alpha \in 
  ({\relA})^\omega$ iff $\alpha$ contains a bad trace. 
  The automaton	$\bbB_1$ has $n$ states and index $k$.
\item\noindent{\hskip-10 pt\bf Step 2:}\ Using a standard construction, see for instance
  \cite{king:comp01}), we transform the automaton into an equivalent
  nondeterministic $\relA$-word automaton $\bbB_2$ with B\"uchi	acceptance
  condition.
  The size of this automaton is bounded by $\mathcal{O}(n k)$. 
\item\noindent{\hskip-10 pt\bf Step 3:}\
   Using the variant of the Safra construction described in
  \cite{pite:nond06}, we transform the automaton $\bbB_2$ into an equivalent
  deterministic parity word automaton $\bbB_3$. The size of $\bbB_3$ is bounded
  by $2^{\mathcal{O}(n k  \log(nk))}$, and $\bbB_3$ has index $\mathcal{O}(nk)$. 
\item\noindent{\hskip-10 pt\bf Step 4:}\
  Finally we define $\bbW$ to be the deterministic parity
  $\relA$-word automaton that accepts the complement of the language of
  $\bbB_{3}$, that is, $\bbW$ accepts a word $\alpha \in ({\relA})^\omega$ iff
  $\alpha$ does \emph{not} contain any bad trace.
  This automaton can be obtained by taking $\bbB_3$ and changing its parity
  function in the same way as in Step 1, i.e., we increase all the parities 
  by $1$. 
  The size and index of $\bbW$ are still bounded by 
  $2^{\mathcal{O}(nk \log(nk))}$ and $\mathcal{O}(nk)$, respectively.
\end{enumerate}
\end{proof}

\begin{theorem}\label{thm:m}
Let $\F$ be a set functor that preserves weak pullbacks, and 
let $\bbA=\struc{A,\ro{a},\De,\Omega}$ be an alternating parity $\F$-automaton
with $n$ states and index $k$. 
Then we can construct an equivalent nondeterministic parity automaton
$\bbA^{\bl}$ of size $2^{\mathcal{O}(n^2 + nk \log(nk))}$ and with index
$\mathcal{O}(nk)$.	
\end{theorem}
\begin{proof}
	In Section~\ref{s:altond} we showed how to transform
	$\bbA$ into an equivalent {\em regular} nondeterministic
	$\F$-automaton ${\bbA}^{\sharp}$. This automaton has 
	at most size
	$2^{n^2}$. Furthermore we can use 
	Proposition~\ref{prop:wordcomp} in order to construct a parity
	$\relA$-word automaton $\bbW$ of size 
	$2^{\mathcal{O}(n k \log(nk))}$ and with index $\mathcal{O}(nk)$
	that accepts exactly those words $\alpha \in ({\relA})^\omega$ 
	that do not contain any bad trace. We define
	$\bbA^{\bl}$ to be the wreath product $\bbA^{\sharp} \odot \bbW$. 
	From Theorem~\ref{t:regaut} we know that 
	$\bbA^{\bl}$ is a nondeterministic $\F$-automaton
	that is equivalent to $\bbA$. Furthermore, spelling out 
        the definitions, one can easily check that 
	$\bbA^{\bl}$ has size
	$2^{\mathcal{O}(n^2 + nk\log(nk))}$ and index 
	$\mathcal{O}(nk)$.
\end{proof}

\begin{remark}
\label{r:fund}
   The complexity bound in our main theorem is  
   the immediate 
   consequence from known 
   results in the literature on $\omega$-automata. 
   In particular we heavily rely
   on~\cite{pite:nond06}. 
   Our contribution
   is to show that 
   known complexity bounds 
   concerning
   $\omega$-automata can be transferred to 
   other types of structures essentially 
   without increasing the complexity. 
   In other words, we prove that
   the complexity of transforming a given alternating 
   $\F$-automaton into an equivalent nondeterministic one
   is bound by the complexity of the Safra construction
   on $\omega$-automata. This observation has been
   further substantiated in~\cite{kiss:deci07} where
   it is shown in detail that both the lower and upper
   complexity bounds from the Safra construction
   on $\omega$-automata yield the respective
   complexity bounds for transforming a given alternating
   tree automaton into an equivalent nondeterministic one.
\end{remark}

\subsection{Closure under union and intersection}

\newcommand{\bbAun}{\bbA_{\cup}}
\newcommand{\bbAit}{\bbA_{\cap}}
\newcommand{\Deun}{\De_{\cup}}
\newcommand{\Deit}{\De_{\cap}}

In this subsection we prove that the class of recognizable languages is
closed under taking unions and intersections. 
This is the content of the following proposition.
Note that in this subsection we do require the 
functor to be standard. As we demonstrated in 
Section~\ref{s:standard} this is not an essential condition. 
In fact, it would be not difficult to modify the
arguments in this subsection in order to show 
closure under union and intersection also for nonstandard functors. 
The fairly simple proofs would, however, look unnecessarily complicated.

\begin{proposition}
\label{prop:closeunion}\label{prop:closeintersec}
Let $\F$ be some set functor.
Given two parity $\F$-automata $\bbA_{1}$ and $\bbA_{2}$ we can construct 
parity $\F$-automata $\bbAun$ and $\bbAit$ such that $L(\bbAun) = L(\bbA_1) \cup
L(\bbA_2)$ and $L(\bbAit) = L(\bbA_1) \cap L(\bbA_2)$. 
Both $\bbAun$ and $\bbAit$ have size $n_1 + n_2 + 1$, where $n_i$ is the size
of $\bbA_i$, and both automata have index $k\mathrel{:=}\max(k_1,k_2)$, where
$k_i$ is the index of $\bbA_i$.
Moreover $\bbAun$ is nondeterministic if $\bbA_1$ and $\bbA_2$ are so.
\end{proposition}

Before we prove the proposition we define the automata
$\bbAun$ and $\bbAit$.

\begin{definition}\label{def:unionit}
Let $\bbA_1 = \struc{A_1, a_{I}^1,\Delta_1, \Om_1}$ and 
$\bbA_2$ $=$ $\struc{A_2, a_{I}^2, \Delta_2, \Om_2}$ be two parity 
$\F$-automata.
We will define their \emdef{sum} $\bbAun$ and \emdef{product} $\bbAit$.

Both of these automata will have the
\emdef{disjoint union} $A_{12} := \{ * \} \uplus A_{1} \uplus A_{2}$
as their collection of states.
Also, the parity function $\Om$ will be the same for both automata:
\[
\Om(a) := 
\left\{ \begin{array}{ll}
0 & \mbox{if } a = *,
\\ \Om_i(a) & \mbox{if } a \in A_{i}.
\end{array} \right.
\]
The only difference between the automata lies in the transition functions,
which are defined as follows:
\[
\Deun(a) := 
\left\{ \begin{array}{ll}
\De_{1}(a^1_I) \cup \De_{2}(a^2_I) & \mbox{if } a = *
\\ \Delta_i(a)                   & \mbox{if } a \in A_{i},
\end{array}\right.
\]
\[
\Deit(a) := 
\left\{ \begin{array}{ll}
\{ \Phi_{1} \cup \Phi_{2} \mid \Phi_{i} \in \De_{i}(\ai^{i}) \}
                & \mbox{if } a = *
\\ \Delta_i(a)  & \mbox{if } a \in A_{i}.
\end{array}\right.
\]

Finally, we put
$\bbAun := \struc{A_{12},a_I,\Deun,\Om}$ and 
$\bbAit := \struc{A_{12},a_I,\Deit,\Om}$.
\end{definition}

Let us now turn to the proof of Proposition~\ref{prop:closeunion}.
\begin{proof}
The automata $\bbAun$ and $\bbAit$ are constructed as defined
above in Definition~\ref{def:unionit}. 
Clearly $\bbAun$ and $\bbAit$ meet the size requirements stated
in the proposition. Furthermore it is easy to see that
$\bbAun$ is nondeterministic if $\bbA_1$ and $\bbA_2$
are nondeterministic automata.
We only show that
$L(\bbAun)=L(\bbA_1) \cup L(\bbA_2)$, 
the other statements of the proposition admit similarly straightforward
proofs. It suffice to show that
$\bbAun$ accepts an arbitrary pointed $\F$-coalgebra
$(\bbS,s)$ iff $\bbA_1$ or $\bbA_{2}$ 
accepts $(\bbS,s)$. Let $(\bbS,s)$ be a pointed $\F$-coalgebra.
and suppose first that the automaton $\bbAun$ accepts $(\bbS,s)$.
Hence by definition, $\eloi$ has a winning strategy $f$ in the game
$\G:=\G(\bbAun,\bbS)$ starting from position $(s,*)$.
Let $i$ be such that $f(*,s) \in \Delta(\ai^{i})$.
It is then straightforward to verify that $f$, restricted to $\eloi$'s
positions in $\G(\bbA_{i},\bbS)$, is a winning strategy for $\eloi$ 
from position $(s,\ai^{i})$.
From this it is immediate that $\bbA_{i}$ accepts $(\bbS,s)$.
Conversely, suppose that $\bbA_{i}$ accepts $(\bbS,s)$, and let $g$ be
a winning strategy for $\eloi$ in the game $\G(\bbA_{i},\bbS)$.
Then in the game $\G(\bbAun,\bbS)$ starting at $(s,*)$, let $\eloi$
start with playing $g(s,\ai^{i}) \in \Deun(*)$, and from then on, play
her strategy $g$.
It is again straightforward to check that this constitutes a winning
strategy for $\eloi$.
\end{proof}

\subsection{Closure under projection}
\label{ss:clpr}

\newcommand{\prF}{\pi_{\F}}

In this subsection we shall prove that the recognizable $\F$-languages are 
closed under {\em existential} projection.
In order to make this notion precise, fix a set functor $\F$ and a set $C$ of
colours.
Recall from Example~\ref{ex:BTK} that we may identify $\F_{C}$-coalgebras
with $C$-coloured $\F$-coalgebras:
Given an $\F_{C}$-coalgebra $\struc{S,\si}$, we may split the coalgebra map $\si: S
\to \F_{C} S$ into two parts, the colouring $\si_{C}:S \to C$ and the
$\F$-coalgebra map $\si_{\F}: S \to \F S$.
Hence, with each $\F_{C}$-coalgebra $\bbS = \struc{S,\si}$ we may associate
its \emph{$\F$-projection} $\prF\bbS := \struc{S,\si_{\F}}$, and likewise for
pointed coalgebras.

A possibility for defining the existential $\F$-projection $\prF L$ of an
$\F_{C}$-language $L$ would be the class
\begin{equation}
\label{eq:nprF}
\{ \prF\bbS \mid \bbS \mbox{ in } L \}.
\end{equation}
Note that we use pseudo-set notation here --- recall that $L$ may be a class
rather than a set.
It is, however, not difficult to see that this language will in general not
be closed under bisimilarity.
If we take the position that bisimilar states represent the \emph{same}
process, this means that (\ref{eq:nprF}) is not the right notion.
This leads to the definition of the existential $C$-projection $\prF L$ of 
$L$ as the closure of (\ref{eq:nprF}) under $\F$-bisimilarity.

\begin{definition}
Let $\F$ be some set functor, let $C$ be some set, and let $L$ be an
$\F_{C}$-language.
The \emdef{existential $\F$-projection} of $L$, notation: $\prF L$, consists
of all pointed $\F$-coalgebras $(\struc{S,\si},s)$ for which there is an
$\F_{C}$-coalgebra $(\struc{S',\ga,\si'},s')$ in $L$ such that
$\struc{S,\si},s \bis_{\F} \struc{S',\si'},s'$.
\end{definition}

In pseudo-set notation we could write
\[
\prF L := \{ (\struc{S,\si},s) \mid
\struc{S,\si},s \bis_{\F} \struc{S',\si'},s' \mbox{ for some }
(\struc{S',\ga,\si'},s') \mbox{ in } L \}.
\]

\begin{remark}
This definition is in accordance with standard usage.
In the case of binary trees, we are dealing with two alphabets $C$ and $D$.
Given a class $K$ of $C\times D$-labeled binary trees, one defines the
\emph{$D$-projection} of this class as the class of $D$-labeled binary
trees $\struc{2^{*},\tau_{D}: 2^{*} \to D}$ for which there is a map
(`$C$-colouring') $\tau_{C}: 2^{*}\to C$ such that the $C \times D$-labeled
binary tree $\tau: 2^{*} \to C \times D$ given by $\tau(s) =
(\tau_{C}(s), \tau_{D}(s))$ belongs to $K$.
No reference to bisimilarity is needed here due to the fact that two
labeled binary trees are bisimilar if and only if they are \emph{identical},
see Example~\ref{ex:bis}.

Second, for Kripke structures our notion of $C$-projection exactly corresponds
to the usual interpretation of  existential bisimulation quantifiers --- a
fact which can be used to prove that closure under projection of 
$\K_{C}$-automata implies uniform interpolation of the modal $\mu$-calculus.
We refer to \textsc{d'Agostino \& Hollenberg}~\cite{agos:logi00} for
more details.
\end{remark}

The main result of this section states that the class of recognizable
languages is closed under this operation.
We will show that, given an $\F_{C}$-automaton $\bbA$, we will define an
$\F$-automaton $\pi_C \bbA$ that accepts a given pointed $\F$-coalgebra
$(\struc{S,\si},s)$ iff there {\em exists} a bisimilar $\F$-coalgebra $(\struc{S',\si'},s')$
and a colouring $\gamma: S' \to C$ such that $(\struc{S',\gamma,\si'},s')$ is
accepted by $\bbA$. 
Before we start to prove this, let us say a word about \emph{universal}
projection. 

\begin{remark}
The universal $\F$-projection $\prF^{\forall}L$ of an $\F_{C}$-language $L$ is
defined dually:
\[
\prF^{\forall} L := \{ (\struc{S,\si},s) \mid
(\struc{S',\ga,\si'},s') \mbox{ in } L
\mbox{ whenever } \struc{S,\si},s \bis_{\F} \struc{S',\si'},s' 
\}.
\]
The question whether the class of recognizable languages is also closed under
universal projection is still open and closely related to the question
whether $\F$-recognizable languages, in general, 
are closed under complementation.
\end{remark}

We now turn to the proof that the recognizable languages are closed under
(existential) projection.
In the remainder of this section, all $\F$-automata are assumed to be
nondeterministic. 
To facilitate the presentation we will think of the transition function
$\Delta$ as a map $A \to \pwE\F A$ and the first component $\Phi$ of a 
strategy $(\Phi,Y)$ for $\eloi$ in an acceptance game $\G(\bbA,\bbS)$
will be regarded as a function of type $A \times S \to \F A$,
cf.~Remark~\ref{r:nd}.

The main result of this subsection is stated in the following proposition.

\begin{proposition}
\label{prop:closeproj}
\label{p:proj}
Let $\F$ be a set functor that preserves weak pullbacks.
For any nondeterministic parity $\F_{C}$-automaton $\bbA$ of size $n$ and
index $k$ we can construct a nondeterministic parity $\F$-automaton $\pi_\F
\bbA$ of size $n$  and index $k$, such that for every pointed $\F$-coalgebra
$(\bbS,s)$ the following are equivalent:
\begin{enumerate}[\em(1)]
\item $\pi_\F \bbA$ accepts $(\bbS,s)$,
\item $\bbA$ accepts an $\F_{C}$-coalgebra $(\struc{
S',\gamma,\sigma'},s')$ 
such that $(\struc{S',\sigma'},s')$ and $(\bbS,s)$ are bisimilar.
\end{enumerate}
\end{proposition}

The remainder of this section is devoted to the proof of this proposition.
First we define the automaton $\pi_C \bbA$ and then we show that it meets the
requirements of the proposition.

\begin{definition}\label{def:automataproj}
Let $C$ be a set, $\bbA = \struc{A,a_I,\Delta,\Om}$ be a parity 
$\F_{C}$-automaton and $\bbA_C = \struc{A,a_I,C,\Delta_C,\Om}$ its $C$-chromatic 
$\F$-companion, see Fact~\ref{f:2:1}. 
Then we define the \emdef{$\F$-projection} $\pi_\F \bbA := \struc{A,a_I,
\Delta_{\pi},\Om}$
where $\Delta_{\pi}(a) := \bigcup_{c \in C} \Delta_C(c,a)$.
\end{definition}
  
\begin{lemma}\label{lem:closeproj1}
If $\bbA$ accepts the $\F_{C}$-coalgebra $(\bbS,s) := (\struc{S,\gamma,\sigma},
s)$ then $\pi_\F \bbA$ accepts $(\bbS^{\pi},s) := (\struc{S,\sigma},s)$.
\end{lemma}

\begin{proof}
The proof is straightforward. 
One has to realize that all the moves of $\eloi$ in the game for $\bbA_C$ are still legitimate moves of $\eloi$
  in the $\pi_\F \bbA$ acceptance game.
\end{proof}
The converse of this lemma however fails in general.

Let $\bbA$ be some $\F_{C}$-automaton and let $(\struc{S,\si},r)$ be a pointed 
$\F$-coalgebra that is accepted by $\pi_\F \bbA$. 
Then we know that $\eloi$ has a winning strategy $(\Phi,Y)$ in $\G(\pi_\F \bbA,
\bbS)$ from position $(r,a_I)$. 
We would like to ensure that $(\Phi,Y)$ is also a winning strategy in 
$\G(\bbA_C,\bbS)$ by defining a coloring $\gamma: S \to C$ as follows: 
$\gamma(s):=c$ if there is a match of $\G(\pi_\F \bbA,\bbS)$, starting from 
position $(r,a_I)$ and conform $\eloi$'s strategy, in which a position 
$(s,a)$ occurs and $\Phi_{s,a} \in \Delta_C(c,a)$.
In general, however, there may be \emdef{distinct} positions $(s,a_1)$
and $(s,a_2)$ that $\abel$ may force the match to pass through, and it may
not be possible to find a single $c \in C$ such that both $\Phi_{s,a_1} \in
\Delta(c,a_1)$ and $\Phi_{s,a_2} \in \Delta(c,a_2)$. 
To avoid this problem we introduce now the notion of \emph{strong}
acceptance.

\begin{definition}
\label{def:scattered}
Let $\bbA$ be a parity $\F$-automaton and $(\bbS,r)$ a pointed $\F$-coalgebra.
A history-free strategy $(\Phi,Y)$ for $\eloi$ in the game $\G(\bbA,\bbS)$ 
initialized at $(r,\ai)$ is called \emdef{scattered} if the relation 
\[
\{ (r,\ai) \} \cup \bigcup
\{ Y_{s,\phi} \sse S \times A \mid (s,\phi) \in \Win_{\eloi} \}
\]
is functional (that is, for every $s \in S$ there is at most one $a \in A$
such that the pair $(s,a)$ belongs to the relation).
Furthermore we say that $\bbA$ {\em strongly accepts} the pointed coalgebra 
$(\bbS,r)$ if $\eloi$ has a scattered winning strategy in the game 
$\G(\bbA,\bbS)$ initialized at position $(r,a_I)$. 
\phantom{blubbblubb}
\end{definition}

As we will see now, strong acceptance is the key to find colorings of 
pointed $\F$-coalgebras.

\begin{lemma}\label{lem:condsuffice}
Let $\bbA$ be an $\F_{C}$-automaton, and let $(\bbS,r)$ be a pointed 
$\F$-coalgebra that is strongly accepted by $\pi_\F \bbA$.
Then there is a $C$-colouring $\gamma:S \to C$ of $\bbS$ such that $\bbA$
accepts $(\struc{S,\gamma,\sigma},r)$.
\end{lemma}

\begin{proof}
Let $(\Phi,Y)$ be a scattered winning strategy for $\eloi$ in 
$\G(\pi_\F \bbA,
\bbS)$. 
According to the definition of scatteredness we can assign to every $s \in S$
a state $a_{s} \in A$ such that $a_{r} = \ai$, and if $(s,a) \in Y_{s,\phi}$ 
for some winning position $(s,\phi)$, then $a = a_{s}$.
Then we define a function $\gamma: S \to C$ as follows.
If there is a $c \in C$ such that $\Phi_{s,a_{s}} \in \Delta_C(c,a)$, then
we pick such a $c$ and put $\gamma(s) := c$; if there is no such $c$, then
we define $\gamma(s) := d$ for some arbitrary $d \in C$.
It follows from these definitions that $(\Phi,Y)$ is a strategy for $\eloi$
in $\G(\bbA_{C},\bbS\oplus\gamma)$ that guarantees her winning every match
starting from $(r,\ai)$.
From this it is immediate that $\bbA$ accepts $(\struc{S,\gamma,\sigma},r)$.
\end{proof}

The next lemma shows that if a pointed coalgebra is accepted by some automaton,
but not strongly so, then we can always find a bisimilar pointed coalgebra that
is strongly accepted.  

\begin{lemma}\label{lem:npower}
Let $\bbA$ be an $\F$-automaton, and let $(\bbS,r)$ be a pointed 
$\F$-coalgebra that is accepted by $\bbA$. 
Then $\bbA$ strongly accepts some pointed $\F$-coalgebra $(\ol{\bbS},\ol{r})$
which is bisimilar to $(\bbS,r)$.
\end{lemma}

\begin{proof}
The coalgebra  $\ol{\bbS}$ will be based on the set $\ol{S} := S \times A$,
and as the selected state $\ol{r}$ of $\ol{\bbS}$ we take the pair $(r,\ai)$.
For the definition of the coalgebra structure $\ol{\si}$, we need some
auxiliary definitions.

First we endow the set $\ol{S}$ with a coalgebra map $\ti{\si}$ such that the
structure $\ti{\bbS} := (\ol{S},\ti{\si})$ is isomorphic to the 
$A$-fold coproduct (`disjoint union') $\coprod_{a \in A} \bbS$. 
For the exact definition of the coproduct of coalgebras 
the reader is referred to \cite{rutt:univ00}.

The canonical injections into $\ti{\bbS}$ are given by the functions
\begin{eqnarray*}
		\kappa_a: S & \to &  S \times A \\
		s & \mapsto & (s,a) 
\end{eqnarray*}  
for all $a \in A$. 

Now consider the first projection map $\pi_{S}: S \times A \to S$.
From $\pi_{S}(s,a) = s$ it follows that 
\begin{equation}
\label{equ:kappa_pi}
	\pi_S \circ \kappa_a = \id_S \qquad \mbox{for all} \; a \in A.
\end{equation}
We are going to prove that 
\begin{equation}
\label{eq:picmor}
\pi_S:(\ol{S},\ti{s}) \to \bbS
\mbox{ is a coalgebra morphism}.
\end{equation}
That is, we will show that the following diagram commutes: \\
\centerline{
\xymatrix{ S \times A \ar[d]_{\ti{\si}} 
      \ar[r]^(.6){\pi_S} & S \ar[d]_{\si} \\
      \F (S \times A) \ar[r]_(.6){\F \pi_S} & \F S}
      }
In order to prove the commutativity of the diagram, take an arbitrary $(s,a)
\in S \times A$.
We obtain the following sequence of identities:
\begin{eqnarray*}
	\F \pi_S (\ti{\si} (s,a)) & = & \F \pi_S (\ti{\si}(
	\kappa_a(s))) \\
	& \stackrel{\mbox{\tiny{$\kappa_a$ coalg. morph.}}}{=} & 
	\F \pi_S( \F \kappa_a (\sigma(s))) \\
	& = & \F (\pi_S \circ \kappa_a)(\si(s)) \\
	& \stackrel{\mbox{\tiny{(\ref{equ:kappa_pi})}}}{=} & 
    \sigma (s) \\
    & = & \sigma(\pi_S (s,a)),
\end{eqnarray*}
which proves (\ref{eq:picmor}).

Second, given a relation $R \sse S \times A$, define the relation $\hat{R} 
\sse \ol{S}\times A$ by putting
\[
\hat{R} := \{ ((s,a),a) \mid (s,a) \in R \}.
\]
Then clearly we have that $R = \converse{\Graph{\pi_{S}}} \circ \hat{R}$,
and hence,
\begin{equation}
\label{eq:3:1}
\Fb R = \converse{\Graph{\F \pi_{S}}} \circ \Fb \hat{R}.
\end{equation}

We are now prepared to turn to the proof of the lemma.
Assume that $(\Phi,Y)$ is a positional strategy for $\eloi$ in the
acceptance game $\G(\bbA,\bbS)$ which is winning for any match starting at
the position $(\ai,r)$.

For the definition of $\ol{\si}: \ol{S} \to \F \ol{S}$, consider an 
arbitrary element $(s,a) \in \ol{S}$, and distinguish cases.
If $(s,a)$ is a \emph{winning} position for $\eloi$ in the game 
$\G(\bbA,\bbS)$, then using (\ref{eq:3:1}), it follows from $(\si(s),\phi) 
\in \Fb Y$, that 
\[ (\si(s),\phi) \in 
  (\converse{\Graph{\F \pi_S}}) \circ \Fb \hat{Y} .\] 
Hence we may take $\ol{\si}(s,a)$ to be some element $x \in \F \ol{S}$ such
that $(\si(s),x) \in (\converse{\Graph{\F \pi_S}})$, that is, $(\F\pi_{S})(x) =
\si(s)$, and $(x,\phi) \in \Fb \hat{Y}$.
If, on the other hand, $(s,a) \not\in \Win_{\eloi}$, then we simply put 
$\ol{\si}(s,a) := \ti{\si}(s,a)$.

We first check that $\pi_{S}$ is indeed an $\F$-coalgebra morphism
from $\ol{\bbS}$ onto $\bbS$.
Take an arbitrary element $(s,a)$ in $\ol{\bbS}$, then we have to check that
$(\F\pi_{S})(\ol{\si}(s,a)) = \si(\pi_{S}(s,a))$.
In case $(s,a) \not\in \Win_{\eloi}$ this follows from the facts that
$\ol{\si}(s,a) = \ti{\si}(s,a)$ and the fact that $\ti{\si}$ is a
coalgebra morphism.
In case $(s,a) \in \Win_{\eloi}$ the identity follows by definition of
$\ol{\si}(s,a)$.

Thus we have proved the first statement of the proposition.
For the second statement, define the strategy $(\ol{\Phi},\ol{Y})$ with
$\ol{\Phi}: \ol{S} \times A \to A$ and $\ol{Y}: \ol{S} \times \F A \to
\Sb{\ol{S}\times A}$ as follows:
\[
\begin{array}{llll}
\ol{\Phi}: & ((s,a),b)    & \mapsto & \Phi_{s,b}
\\
\ol{Y}:    & ((s,a),\phi) & \mapsto & \hat{Y}_{s,\phi}.
\end{array}
\]
Since all relations chosen by $\eloi$ are of the form $\hat{R}$, and all
elements of such relations are of the form $((s,a),b)$ with $a = b$, it
is obvious that the set $\{ ((s,\ai),\ai) \} \cup \bigcup \{ 
\hat{Y}_{s,\phi} \mid 
(s,\phi) \in \Win_{\eloi} \}$ is functional.
In other words, the strategy is scattered.

Thus it is left to prove that $(\ol{\Phi},\ol{Y})$ guarantees $\eloi$ to win
any match of $\G(\bbA,\ol{\bbS})$ starting from $(\ol{r},\ai)$.
To see why this is the case, consider an arbitrary position $((s,a),a)$ with
$(s,a) \in \Win_{\eloi}(\G(\bbA,\bbS))$, and abbreviate $\phi := \Phi_{s,a}$.
Then by definition, $\ol{\Phi}((s,a),a) = \phi$ and $\ol{Y}((s,a),\phi) = 
\hat{Y}_{s,\phi} = \{ ((t,b),b) \mid (t,b) \in Y_{s,\phi} \}$.
From this observation it is easy to derive that for any $\G(\bbA,\ol{\bbS})$
match $(\ol{r},\ai)((s_{1},a_{1}),a_{1})((s_{2},a_{2}),a_{2}) \ldots$ that is
conform the strategy $(\ol{\Phi},\ol{Y})$, the corresponding $\G(\bbA,\bbS)$ 
match $(r,\ai)(s_{1},a_{1})(s_{2},a_{2})\ldots$ is conform $(\Phi,Y)$.
And since this strategy was supposed to be winning for $\eloi$ from $(r,\ai)$, 
it follows that the $\G(\bbA,\ol{\bbS})$ match is, indeed, a win for $\eloi$.
This proves the second statement of the proposition.
\end{proof}

We are now ready to prove our main result.

\begin{proofof}{Proposition~\ref{prop:closeproj}}
The implication ($1 \Rightarrow 2$) is immediate by the Lemmas~\ref{lem:npower} 
and~\ref{lem:condsuffice}.
The other implication follows from Lemma~\ref{lem:closeproj1} and the
observation~\cite{vene:auto04} that $\F$-automata do not distinguish between 
bisimilar pointed $\F$-coalgebras.
\end{proofof}

Together with Theorem~\ref{thm:m} the proposition entails what we call closure
under (existential) projection.
\begin{corollary}
Let $\F$ be a set functor that preserves weak pullbacks.
Given an alternating parity $\F_{C}$-automaton $\bbA$ of size $n$ and
index $k$ we can construct a nondeterministic parity $\F$-automaton 
$\pi_\F \bbA^{\bl}$ such that the following are equivalent:
\begin{enumerate}
\item $\pi_\F\bbA^{\bl}$ 
accepts a pointed $\F$-coalgebra $(\struc{S,\si},\ro{s})$,
\item $\bbA$ accepts a pointed $\F_{C}$-coalgebra 
$(\struc{S',\gamma,\si'},\ro{s}')$
such that $(\struc{S,\si},\ro{s})$ 
and $(\struc{S',\si'},\ro{s}')$ are $\F$-bisimilar.
\end{enumerate}
The size of $\pi_\F \bbA^{\bl}$ is 
$2^{\mathcal{O}(n^2 + nk\log(nk))}$ and the
index of $\pi_\F \bbA^{\bl}$ is $\mathcal{O}(nk)$.
\end{corollary}

\section{Solution of the nonemptiness problem}
\label{s:6}

\noindent In this section we prove that every parity $\F$-automaton
accepts a \emph{finite} coalgebra, if it accepts a coalgebra at all.
The key result leading to this observation is that any
nondeterministic parity automaton with a nonempty language, actually
accepts a coalgebra that `lives inside the automaton', in the
following sense.

\begin{theorem}\label{thm:nonempti}
Let $\F$ be some weak pullback preserving set functor, and let
$\bbA = \struc{A,\ro{a},\De,\Omega}$ 
be a nondeterministic parity $\F$-automaton.
Then $\bbA$ accepts some pointed $\F$-coalgebra iff $\bbA$ accepts a pointed
$\F$-coalgebra $(\struc{S,\si},s_0)$ 
with $S \subseteq A$, $s_0=\ro{a}$ and $\si(s) 
\in \De(s)$ for all $s \in S$.
\end{theorem}

As an immediate consequence of the above result, and of the fact that for
every alternating $\F$-coalgebra automaton we can effectively construct an
equivalent nondeterministic automaton (Theorem~\ref{thm:m}), we obtain the
following solution for the nonemptiness problem for parity $\F$-automata.

\begin{corollary}
\label{c:fmp}
Let $\F$ be some weak pullback preserving set functor, and let $\bbA = 
\struc{A,\ro{a},\De,\Omega}$ be a parity $\F$-automaton of size $n$. 
Then $\Lan(\bbA) \not= \nada$ iff $\bbA$ accepts a pointed $\F$-coalgebra
$(S,\si,s_0)$ with $\card{S} \leq 2^{\mathcal{O}(n^2 \log n)}$.
\end{corollary}

\begin{proof}
Suppose $\Lan(\bbA) \not= \nada$, i.e.\ there is some pointed 
$\F$-coalgebra that is accepted by $\bbA$.
The index of $\bbA$ is smaller than or equal to its size $n$, and so it follows
from Theorem~\ref{thm:m} that we can transform $\bbA$ into an equivalent
nondeterministic parity automaton $\bbA^{\bl}$ of size
$2^{\mathcal{O}(n^2 + n^2 \log(n^2))}= 2^{\mathcal{O} (n^2 \log n)}$. 
Because $\bbA^{\bl}$ is equivalent to $\bbA$ we know that $\bbA^{\bl}$ accepts
some pointed $\F$-coalgebra.
The claim follows now immediately from Theorem~\ref{thm:nonempti}.
\end{proof}

The remainder of this section is devoted to the proof of
Theorem~\ref{thm:nonempti}.

\begin{remark}
%
   The heart of the proof of Theorem~\ref{thm:nonempti} will be
   the construction of the so-called nonemptiness 
   game of a nondeterministic $\F$-automaton. This
   nonemptiness game can be seen as a variant of the acceptance
   game of an $\F$-automaton that is played {\em inside} the given
   automaton. Readers familiar with automata that operate
   on infinite objects will recognize an analogy to standard
   automata theoretic techniques: in order to solve the 
   nonemptiness problem of a given nondeterministic 
   automaton one usually considers an {\em input-free} variant
   of the automaton and then decides whether 
   the input-free automaton has a successful run 
   (cf.~e.g.~\cite[Chap.~8]{grae:auto02} concerning tree automata). 
   Successful runs
   of such an input-free automaton correspond to
   winning strategies of $\elo$ in our nonemptiness game.    
\end{remark}
\subsection{Standardness and the nonemptiness game}

For a smooth presentation of the proof of Theorem~\ref{thm:nonempti}, we
first consider the special case, where we impose the additional condition 
on the functor to be \emph{standard} (cf.\ Section~\ref{s:standard}).
This means that in particular, $S \sse T$ implies $\F S \sse \F T$.
Furthermore we need the following properties of standard functors.

\begin{proposition}
\label{f:st}
Let $\F$ be a standard, weak pullback preserving functor. 
Then, for all sets $S$, $T$, $S'$ and $T'$, with $S' \sse S$ and $T' \sse T$, 
and for all relations $R \sse S \times T$:

\noindent
(1) $\F$ commutes with intersections: $\F (S \cap T) = \F S \cap \F T$.

\noindent
(2) $\Fb$ commutes with restrictions: 
$\Fb (R\bpr_{S'\times T'}) = (\Fb R)\bpr_{\F S' \times \F T'}$;
\end{proposition}

\begin{proof}
By a result of Trnkov\'a (see~\cite{trnk:some69}), property (1) holds for 
any standard functor $\F$, provided that the intersection $S \cap T$ is
\emph{nonempty}.
We use weak pullback preservation of $\F$ to show that the claim is true for
arbitrary intersections.
Let $S$ and $T$ be sets.
The left diagram below is a pullback diagram which, by our assumption that
$\F$ is standard, gets mapped to the lower right square:
\[\begin{array}{ccc} 
\xymatrix{ 
      S \cap T \ar[r]^(.6){\iota_{S\cap T,S}} \ar[d]_{
      \iota_{S\cap T,T}} & S 
      \ar[d]^{\iota_{S,S\cup T}}\\
       T \ar[r]_(.4){\iota_{T,S\cup T}} & S \cup T } & \qquad \quad
\xymatrix{ 
      \F S \cap \F T \ar@/_/[ddr]_{\iota_{\F S\cap \F T, \F T}}
      \ar@/^/[drr]^{\iota_{\F S\cap \F T, \F S}} 
      \ar@{-->}[dr]^h & & \\
      & \F(S \cap T) \ar[r]_(.55){\iota_{\F(S\cap T),\F S}} 
      \ar[d]^{\iota_{\F(S\cap T),\F T}} & \F S 
      \ar[d]^{\iota_{\F S,\F (S \cup T )}}\\
      & \F T \ar[r]_(.4){\iota_{\F T,\F (S \cup T )}} & \F(S \cup T) }
\end{array}\]
From the fact that $\F$ preserves weak pullbacks, it follows that the square
in the right diagram is a weak pullback diagram.
Hence there exists some function $h$ that makes both upper triangles 
commute as depicted in the diagram. 
But then a straightforward verification shows that $h$ itself must be the
inclusion from $\F S \cap \F T$ into $\F(S \cap T)$, i.e.\ $\F S \cap \F T
\sse \F (S \cap T)$. 
The converse inclusion $\F  (S \cap T) \sse \F S \cap \F T$ is an immediate
consequence of standardness of the functor $\F$.

For (2), we first consider the inclusion $\sse$.
By monotonicity of $\Fb$ we have $\Fb (R\bpr_{S'\times T'}) \sse \Fb (R)$.

From this it follows immediately that
\[ 
\Fb (R\bpr_{S'\times T'}) \sse
(\Fb R)\bpr_{\F S' \times \F T'} 
\] 
A short proof of the opposite inclusion can be found
in~\cite[Prop.~2.2]{vene:auto06}.
\end{proof}

The key concept in our proof of Theorem~\ref{thm:nonempti} is the so-called
\emph{nonemptiness game} $\nonempty{\bbA}$ that we may associate with a
nondeterministic automaton.
Intuitively, one should think of this game as the simultaneous projection on
$\bbA$ of all acceptance games $\G(\bbA,\bbS)$.
For a formal definition we need the following notion.

\begin{definition}
Let $\F$ be a standard set functor.
Given a finite set $A$ and an element $\phi \in \F A$, the set 
\[ 
\Base(\phi) \coloneqq
  \bigcap \{U \mid U \subseteq A \mbox{ and } \phi \in \F U\}. 
\]
is defined as the \emdef{base} of $\phi$.
\end{definition}

It follows from $\phi \in \F A$ that the set $\{U \mid U \subseteq A$ and
$\phi \in \F U\}$ is nonempty, so that $\Base(\phi)$ is well defined.

\begin{example}
Fix some finite set $A$, and recall the definition of the functors $\BT$ and
$\K$ of Example~\ref{ex:BTK}.
The base of an arbitrary element $(a_1,a_2) \in \BT A$ is the set $\{ a_1,
a_2\}$.
An element of $\K X$ is a subset $B\subseteq A$; the base of such an element
is the set $B$ itself.
\end{example}

Intuitively the base of an element $\phi \in \F A$ consists exactly of those
elements of $A$ that we need to `construct' $\phi$.
Bases have the following key property.

\begin{proposition}
Let $\F$ be a standard set functor, and consider an object $\phi \in \F A$,
where $A$ is some finite set.
Then $\Base(\phi)$ is the smallest set $X$ such that $\phi \in \F X$.
\end{proposition}

\begin{proof}
It is an easy consequence of Fact~\ref{f:st} and the finiteness of $A$ that
$\phi \in \F (\Base(\phi))$.
Now suppose $Y$ is a set such that $\phi \in \F Y$. 
Then we have $\phi \in \F A \cap \F Y = \F (A \cap Y)$, so that
$\Base(\phi)$ is a subset of $A \cap Y$, and, hence, of $Y$. 
\end{proof}

In the remaining part of this subsection we assume a fixed standard functor
$\F$.
We can now define the `nonemptiness game' $\nonempty{\bbA}$ associated with
a given nondeterministic parity $\F$-automaton $\bbA$.

\begin{definition}
Let $\bbA= \struc{A,\ro{a},\Delta, \Omega}$ be a nondeterministic parity
$\F$-automaton, where $\F$ is a standard set functor.
The rules and the (parity) winning conditions of the \emdef{nonemptiness game}
$\nonempty{\bbA}$ of $\bbA$ are given in Table~\ref{tb:ne}. 
\begin{table*}[t]
\begin{center} 
\begin{tabular}[b]{|l|c|c|c|} \hline
Position: $b$ & Player & Admissible moves: $E[b]$ & $\Omega'(b)$ \\ \hline
$ a \in A $   & $\elo$ &   $ \De(a) $             & $\Omega(a)$  \\
$\phi \in \F A$ & $\elo$ & $\{\Base(\phi)\}$      & 0 \\
$B \subseteq A$ & $\abe$ & $B$          & 0 \\ \hline
\end{tabular}
\end{center}
\caption{Nonemptiness game of a nondeterministic $\F$-automaton}
\label{tb:ne}
\end{table*}
\end{definition}

For an informal description of this game, we first note that, just like the
acceptance game, matches proceed in rounds.
The basic positions of the game are now the states of the automaton $A$.
At a basic position $a \in A$ player $\elo$ has to move to some successor
$\phi \in \De (a)$. 
Then $\elo$ moves further to the base of $\phi$. 
Finally it is $\abe$'s turn to chose some $a' \in \Base(\phi)$ as the next
basic position. 

Attentive readers may have noticed that the formulation of $\nonempty{\bbA}$
looks unnecessarily complicated because $\elo$'s second move (from $\phi\in
\F A$ to $\Base(\phi) \in \power A$) is entirely determined by her first move. 
We keep this redundancy because it makes it easier to relate matches of the
nonemptiness game to matches of the acceptance game of $\bbA$.

The name `nonemptiness game' can be justified by the following two lemmas
that, taken together, imply that $\eloi$ has a winning strategy in the
nonemptiness game for $\bbA$ iff the language recognized by $\bbA$ is not
empty.
The first lemma takes care of the direction from left to right (recall that
$\nonempty{\bbA}$, being a parity game, satisfies history-free determinacy).

\begin{lemma}\label{lem:emptygame}
Let $\bbA= \struc{A,\ro{a},\Delta,\Om}$ be a nondeterministic $\F$-automaton.
If $\varphi: A \to \F A$ encodes a history-free winning strategy of $\elo$
in $\nonempty{\bbA}$ at position $\ro{a}$ then $\bbA$ accepts the 
$\F$-coalgebra $(A,\varphi,\ro{a})$.
\end{lemma}

\begin{proof}
Let $\varphi:A \to \F A$ be a winning strategy of $\elo$ in $\nonempty{\bbA}$
and define $\bbA_\varphi$ to be the $\F$-coalgebra $\struc{A,\varphi}$. 
In order to show that $\bbA$ accepts the pointed coalgebra $(\bbA_\varphi,
\ro{a})$ we have to equip $\elo$ with a winning strategy in $\game=
\game(\bbA,\bbA_\varphi)$. To this aim, we define:
\[
\begin{array}{rcccrcccc}
  	\Phi: & \Id_A  & \to & \F A  & \qquad  & Z: & A \times \F A
		& \to & \Pow (\Id_A) \\
		& (a,a) & \mapsto & \varphi(a) & & & 
		(a,\psi) & \mapsto & \quad \Id_{\Base(\psi)} 
\end{array}
\]
The functions $\Phi$ and $Z$ encode a legitimate strategy for $\elo$ in
$\game$ at position $(\ro{a},\ro{a})$: in order to see this, first note that
any match that starts at $(\ro{a},\ro{a})$ and in which $\elo$ plays conform
$(\Phi,Z)$ will only pass through basic positions of the form $(a,a)$.
Hence, $\elo$'s strategy is defined on all positions that are possibly
reached in such a match.
Let us now see that at any position of the form $(a,a)$ for some $a \in A$,
the moves encoded by $\elo$'s strategy are legitimate. 
That $\varphi(a)$ is an element of $\De(a)$ is true by definition. 
In order to see that the move further to $\Id_\Base(\varphi(a))$ is also a
legal move we use Fact~\ref{fact:wpb}(2) which yields that 
$\Fb(\Id_{\Base(\varphi(a))}) = \Id_{\F \Base(\varphi(a))}$. 
Together with $\varphi(a) \in \F \Base(\phi(a))$ this implies that
$\varphi(a), \varphi(a)) \in \Fb(\Id_{\Base(\varphi(a))})$.
  
It remains to show that $(\Phi,Z)$ is indeed a {\em winning} strategy for
$\elo$. 
The key observation here is that the "projection" of a $\G$-match in which
$\elo$ plays conform $(\Phi,Z)$ is a match of $\nonempty{\bbA}$ in which
$\elo$ plays conform $\phi$:
\[
\begin{array}{lccccccc}
\mbox{$\G$-match} & \quad & (\ro{a},\ro{a}) & (\phi(a),\ro{a}) &
  \Id_{\Base(\phi(a))} & (a_1,a_1) & \ldots & (a_n,a_n)  \\ 
& & \Downarrow & \Downarrow & \Downarrow & \Downarrow &
  \phantom{vdots} & \Downarrow \\
\mbox{projection} & & \ro{a} & \phi(a) & \Base(\phi(a)) & a_1 & \ldots & a_n 
\end{array}
\]
This suffices to prove that $\elo$ wins all $\G$-matches in which she follows
the strategy $(\Phi,Z)$, because we assumed $\phi$  to be a winning strategy
for $\elo$ in $\nonempty{\bbA}$. 
\end{proof}

The next lemma states that if the language recognized by $\bbA$ is not
empty, then $\eloi$ wins the nonemptiness game of $\bbA$ indeed.

\begin{lemma}\label{lem:nonempty}
Let $\bbA = \struc{A,\ro{a},\Delta,\Om}$ be a nondeterministic $\F$-automaton.
If $\bbA$ accepts some pointed $\F$-coalgebra $(\bbS,s_0)$ then $\elo$ has
a winning strategy in $\nonempty{\bbA}$ at position $\ro{a}$.
\end{lemma}

\begin{proof}
Suppose $\bbA$ accepts the pointed $\F$-coalgebra $(\bbS,s_0) =
(\struc{S,\sigma},s_0)$. 
Then $\elo$ has a history-free winning strategy in the acceptance game $\game
= \game(\bbA,\bbS)$ starting at position $(s_0,\ro{a})$, that can be
encoded as a pair of functions
\[ 
\left(\Phi: S \times A \to \F A,
Z: S \times \F A \to \Pow (S \times A) \right).
\]
Moreover, without loss of generality we may assume that 
\begin{equation}
\label{eq:base}
\range(Z_{s,\phi})= \Base(\phi)
\end{equation}
for all positions $(s,\phi)$ that are winning for $\eloi$.
To see this, observe that for all $(s,\phi) \in \Win_{\eloi}$ we have 
$(\si,\phi) \in \Fb(Z_{s,\phi})$ by legitimacy of the strategy.
Thus it follows from $\phi \in \F \Base(\phi)$ that 
\[
(\si(s),\phi) \in \left(\Fb 
	Z_{(s,\phi)}\right)\rst{\F S \times \F \Base(\phi)}.
\]
But Fact~\ref{fact:wpb} yields
\[
\left(\Fb Z_{(s,\phi)}\right)\rst{\F S \times \F \Base(\phi)}
\;=\;
\Fb \left(Z_{(s,\phi)}\rst{S \times \Base(\phi)}\right),
\]
so that we may infer that  
\[
(\si(s),\phi) \in
\Fb \left(Z_{(s,\phi)}\rst{S \times \Base(\phi)}\right).
\]
From this it follows that instead of playing $Z_{s,\phi}$, $\eloi$ could
have played the relation $Z_{(s,\phi)}\rst{S \times \Base(\phi)}$ as well.
Since the latter relation is \emph{smaller} it decreases the choice of
$\abel$, and so $\eloi$ will increase rather than decrease her chances
of winning the game.
This shows that indeed we may assume (\ref{eq:base}) without loss of 
generality.

We now turn to the nonemptiness game $\nonempty{\bbA}$, and show that the
strategy $(\Phi,Z)$ can be used to equip $\elo$ with a winning strategy.
Call a position $(s,a)$ in a $\G(\bbA,\bbS)$-match {\em parallel} to $a'$ if
$a=a'$ and $(s,a)$ is winning for $\eloi$.

We first describe $\elo$'s strategy in one round of the game $\nonempty{\bbA}$,
and demonstrate how she constructs a parallel round of $\G(\bbA,\bbS)$.
Let $a$ be a position in a $\nonempty{\bbA}$-match and let $(s,a)$ be an
(inductively defined) parallel position in $\G(\bbA,\bbS)$. 
$\elo$'s strategy is to move from $a$ to $\phi := \Phi_{(s,a)}$ and further to 
$\Base(\phi)$.
After that $\abe$ chooses an element $a'$ of $\Base(\phi)$.  
The corresponding round of the  $\G(\bbA,\bbS)$-match is constructed as
follows: $\elo$ moves from $(s,a)$ to $(s,\phi)$ and further to 
$Z_{(s,\phi)}$. 
Now we use our assumption (\ref{eq:base}): from $\range(Z_{(s,\phi)}) =
\Base(\phi)$ and $a' \in \Base(\phi)$ we may infer the existence of some $s'$
such that $(s',a') \in Z_{(s,\phi)}$.
Therefore in $\G(\bbA,\bbS)$ $\abel$ can move from $Z_{(s,\phi)}$ to $(s',a')$.
This position $(s',a')$ is parallel to $a'$ because $\elo$ played according 
to her winning strategy in $\G(\bbA,\bbS)$, and so $(s',a')$ is winning for
$\eloi$.

It should then be obvious how this strategy leads to a victory for $\eloi$
in the nonemptiness game.
The $\nonempty{\bbA}$-match $\pi$ starts at position $\ro{a}$ and the
$\G(\bbA,\bbS)$-match $\pi'$ starts at the parallel position $(s_0,\ro{a})$.
Now if $\eloi$ plays the strategy sketched above, then for any resulting
$\nonempty{\bbA}$-match 
\[ 
\pi = \ro{a} \ldots a_1 \ldots a_2 \ldots a_3 \ldots 
\]
there is a parallel $\G(\bbA,\bbS)$-match 
\[ 
\pi' = (s_0,\ro{a}) \ldots (s_1,a_1) \ldots (s_2,a_2) \ldots
  (s_3,a_3) \ldots 
\]
which is conform her winning strategy $(\Phi,Z)$. 
From this it follows immediately that $\eloi$ wins the match $\pi$.
\end{proof}

In the case of a \emph{standard} functor, Theorem~\ref{thm:nonempti} is an
almost immediate consequence of the above two lemmas, together with the
history-free determinacy of parity games.
For the nontrivial direction of the theorem, suppose that $\bbA$ accepts some
pointed $\F$-coalgebra. 
So, by Lemma~\ref{lem:nonempty}, $\elo$ has a winning strategy in the 
nonemptiness game $\nonempty{\bbA}$ starting at position $\ro{a}$.
As parity games are history-free determined this implies $\elo$ actually has 
a \emph{positional} winning strategy $\varphi$ from position $\ro{a}$ in 
$\nonempty{\bbA}$.
Then Lemma~\ref{lem:emptygame} implies that $\bbA$ accepts $(A,\varphi,a_I)$
which finishes the proof of the theorem. 
For the general case, we have to do a little more work.

\subsection{Nonemptiness problem: the general case} 

The idea underlying the proof of Theorem~\ref{thm:nonempti} is simply a 
reduction of the general situation to the standard case.
For this purpose we need to define the \emph{standardization} of an
$\F$-automaton.

\begin{definition}
Let $\F$ be a weak pullback preserving set functor, and suppose that
$\Fl$ is a standardization of $\F$, i.e.\ there is a natural
isomorphism $\la: \F \cong \Fl$ (see Definition~\ref{d:stand}).
Given a nondeterministic $\F$-automaton $\bbA=\struc{A,\ro{a},\De,\Acc}$,
we first define the map $\stl{\De}: A \to \Fl A$ by putting
\[ 
\stl{\De}(a) := \la_{A}[\De(a)].
\]
Then the \emdef{$\la$-standardization} of $\bbA$ is the automaton 
$\stl{\bbA} \coloneqq \struc{A,\ro{a},\stl{\De},\Acc}$.
\end{definition}

In words, $\stl{\De}(a)$ is the direct image of $\De(a)$ under the bijection
$\la_A$.  

\begin{proposition}
\label{p:staut}
Let $\F$ be a weak pullback preserving set functor, suppose that 
$\la: \F \cong \Fl$, and let $\bbA=\struc{A,\ro{a},\De,\Acc}$ be a 
nondeterministic $\F$-automaton.
Then $\bbA$ accepts a pointed $\F$-coalgebra $\struc{S,\si,s}$ iff
$\stl{\bbA}$ accepts the pointed $\Fl$-coalgebra $\struc{S,\la_S\circ \si,s}$.
\end{proposition}

\begin{proof}
Let $\bbS= \struc{S,\sigma}$ and let $\stl{\bbS} = \struc{S,\stl{\si}}$
be the $\Fl$-coalgebra given by $\stl{\si} = \la_{S}\circ\si$.
We prove the implication from left to right, the converse direction can be
proven analogously.

Suppose that $\bbS$ is accepted by $\bbA$. 
We want to show that $\stl{\bbS}$ is accepted by $\stl{\bbA}$.
By assumption, $\elo$ has a positional strategy 
$(\Phi:S \times A \to \F A,Z:S \times \F A \to S\times A)$
in $\G := \G(\bbA,\bbS)$, which is winning for any match starting at position
$(s,\ro{a})$.

Now define the following positional strategy 
$(\stl{\Phi} : S \times A \to \Fl A, Z: S \times \Fl A \to S \times A)$
for $\eloi$ in $\stl{\G} = \G(\stl{\bbA},\stl{\bbS})$:
\[
\begin{array}{lll}
   \stl{\Phi}(s,a) & := & \la_{A}\circ\Phi(s,a),
\\ \stl{Z}(s,\al)  & := & Z(s, \la_{A}^{-1}(\al)).
\end{array}
\]
Here $\la_{A}^{-1}: \Fl A \to \F A$ denotes the \emph{inverse} of the
bijection $\la_{A}$.

It is obvious, that for any basic position $(s,a)$, after two moves in
$\stl{\G}$, $\eloi$ arrives  at the \emph{same} binary relation
$Z(s,\Phi(s,a))$, as after two moves in $\G$.
From this it follows immediately, that $(\stl{\Phi},\stl{Z})$ is a winning
strategy for any $\stl{\G}$-match starting at position $(s,\ai)$.

The only thing that is left to prove is that $(\stl{\Phi},\stl{Z})$ only
suggests \emph{legitimate moves}, at least, when we start at a basic
positions that is winning for $\eloi$.
The only case worth worrying about, is whether $\stl{Z}$ is legitimate at 
position $(s,\la_{A}(\phi)) \in S \times \Fl A$, if $Z$ is legitimate at
position $(s,\phi) \in S \times \F A$.
That is, we assume that $(\si(s),\phi) \in \Fb Z$, and need to show that 
$(\stl{\si}(s),\la_{A}(\phi)) \in \Flb \stl{Z}$.

By definition of relation lifting, it follows from $(\si(s),\phi) \in \Fb Z$
that there is some $\zeta \in \F Z$ such that 
\[
\begin{array}{lll}
   \si(s) & := & (\F \pi_{S})(\zeta),
\\ \phi   & := & (\F \pi_{A})(\zeta),
\end{array}
\]
where $\pi_{S}: Z \to S$ and $\pi_{A}: Z \to A$ are the projections.
Furthermore the following diagram commutes by naturality of $\lambda$:
\[ \xymatrix{
\F S \ar[d]_{\la_S}  & \F Z  \ar[d]_{\la_Z } \ar[l]_{\F \pi_{S}} 
    \ar[r]^{\F \pi_{A}} & \F A  \ar[d]_{\la_A} 
\\    
\Fl S & \Fl Z \ar[l]^{\Fl \pi_{S}} \ar[r]_{\Fl \pi_{A}} & \Fl A
}\]
Hence, for $\stl{\zeta} := \la_{Z}(\zeta) \in \Fl Z$ we may compute:
\begin{eqnarray*}
(\Fl\pi_{S})(\stl{\zeta}) &=& ((\Fl\pi_{S})\circ\la_{Z})(\zeta)
\\ &=& (\la_{S}\circ\F\pi_{S})(\zeta)
\\ &=& \la_{S}(\si(s))
\\ &=& \stl{\si}(s),
\end{eqnarray*}
and
\begin{eqnarray*}
(\Fl\pi_{A})(\stl{\zeta}) &=& ((\Fl\pi_{A})\circ\la_{Z})(\zeta)
\\ &=& (\la_{A}\circ\F\pi_{A})(\zeta)
\\ &=& \la_{A}(\phi).
\end{eqnarray*}
From this it is immediate, by definition of $\Flb$, that
$(\stl{\si}(s),\la_{A}(\phi)) \in \Flb \stl{Z}$.
\end{proof}

The proof of the main result in this section now follows easily.

\begin{proofof}{Theorem~\ref{thm:nonempti}}
For the nontrivial direction of the Theorem, let $\bbA = 
\struc{A,\ro{a},\De,\Om}$ be an nondeterministic parity $\F$-automaton with
$\Lan(\bbA) \neq \nada$.
As an immediate consequence of Proposition~\ref{p:staut}, we see that 
$\Lan(\stl{\bbA}) \neq \nada$, where $\Fl$ is some standardization of $\F$
via a natural isomorphism $\la$.
Since we already showed the theorem to hold for standard functors, this means
that there is some $\Fl$-coalgebra $\struc{A,\phi}$, with $\phi(a) \in
\stl{\De}(a)$ for all $a\in A$, such that $(\struc{A,\phi},\ai)$ is
accepted by $\stl{\bbA}$.

Now consider the $\F$-coalgebra $\struc{A,\al}$ given by defining 
\[
\al(a) := \la_{A}^{-1}(\phi(a)),
\]
where $\la_{A}^{-1}$ is the inverse of the bijection $\la_{A}$.
Clearly then, $\phi = \la_{A}\circ\al$, and so it follows from
Proposition~\ref{p:staut}, that $\bbA$ accepts $(\struc{A,\al},\ai)$.
Finally, by definition of $\stl{\De}$, $\phi(a) = \la_{A}(\al'(a))$ for
some $\al'(a) \in \De(a)$.
Since $\la_{A}$ is a bijection, it follows that $\al(a) = \al'(a)$, so
that $\al(a) \in \De(a)$, as required.
\end{proofof}

\section{Conclusions \& Questions}
\label{s:concl}

\noindent There is a long list of issues that need some further discussion. 
To start with, we believe that this paper provides evidence for the claim 
that universal coalgebra forms an appropriate \emph{abstraction level} for 
studying automata theory.
Our results show that important automata-theoretic phenomena have a natural 
existence at the coalgebraic level of abstraction.

Second, although we have hardly mentioned logic at all, the results in the 
paper have in fact significant logical corollaries.
For instance, given the connection between formulas of coalgebraic fixed-point
logics and coalgebra automata theory, established
in~\cite{vene:auto04,vene:auto06}, it is easy to show that the logics 
introduced in the mentioned work, have the finite model property.
Or, generalizing results in~\cite{agos:logi00}, we can show that the
coalgebraic fixed-point logics of~\cite{vene:auto04} all have some
kind of \emph{uniform interpolation}.
We hope to say more on this in future work.

Probably the most important issue to be addressed concerns the closure of
the class of recognizable languages under \emph{complementation}.
For our coalgebraic automata it is not so easy to prove a
complementation lemma, even for alternating or deterministic automata.
The reason for this is that the acceptance game for coalgebraic automata has
some crucial \emph{nonsymmetric} interaction between the two players, with 
$\eloi$ choosing relations and $\abel$ picking elements of such relations.
The fact that for many well-known functors (including the ones that yield
simple coalgebras such as trees and transition systems), this game can be 
brought into a symmetric form, simply reveals the existence of an interesting 
\emph{property} that some functors have, and others may not.
We have to leave this matter as an intriguing area for further research,
however.
Should there be a strong need for closure of recognizable languages under 
complementation, one may always consider to move to a different notion of 
coalgebra automaton that is tailored towards a more symmetric acceptance 
game.
This is also a matter that we leave for future investigations.

In any case, closure under complementation may be a less important property 
than it appears to be at first sight.
Explained in logical terms, the point is that coalgebraic logics (with or
without fixed-point operators) \emph{without negation} already have
considerable expressive power.
For instance, A.~Baltag (private communication) has shown that any state in a 
\emph{finite} coalgebra can be completely characterized (modulo bisimilarity)
by a negation free coalgebraic fixed-point formula, see Corollary~7.1 
in~\cite{vene:auto06}.

In order to gain a better understanding of coalgebra automata it will also 
be useful to investigate instances of $\F$-coalgebra automata other than word,
tree or graph automata. 
In Definition~\ref{d:ekpf} we defined a class of $\Set$-functors which all
fall into the scope of our work.
Future research will show whether e.g.\ the coalgebra automata for the 
functors $\probdist$ and $\multiset$ yield reasonable automata for 
probabilistic transition systems and for directed weighted graphs respectively.
 
Finally, we are interested to see whether the conditions on the functor are
really needed. 
We believe that our main result crucially depends on the fact the functor
preserves weak pullbacks.
This is in line with results by Trnkov\'a~\cite{adam:auto90} indicating
that for a related class of functorial automata, nondeterministic and 
deterministic recognizability coincide if \emph{and only if} the functor
preserves weak pullbacks.
The precise connection with these results clearly needs to be investigated.

\bibliographystyle{plain}
\bibliography{automata}
\end{document}